\documentclass[aps,pra,preprint,notitlepage,superscriptaddress,nofootinbib,floatfix,showkeys]{revtex4-2}

\usepackage[utf8]{inputenc}
\usepackage[T1]{fontenc}
\usepackage{amsmath,amssymb,amsthm}
\usepackage{graphicx}

\newtheorem{theorem}{Theorem}
\newtheorem{proposition}{Proposition}
\newtheorem{lemma}{Lemma}

\newtheorem{corollary}{Corollary}
\newtheorem{definition}{Definition}
\newtheorem{assumption}{Assumption}
\newtheorem{remark}{Remark}

\newcommand{\Tr}{\mathrm{Tr}}
\newcommand{\E}{\mathbb{E}}
\newcommand{\Var}{\mathrm{Var}}
\newcommand{\Cov}{\mathrm{Cov}}
\newcommand{\MSE}{\mathrm{MSE}}
\newcommand{\argmin}{\operatorname*{arg\,min}}

\newcommand{\norm}[1]{\lVert #1 \rVert}
\newcommand{\ket}[1]{|#1\rangle}
\newcommand{\bra}[1]{\langle#1|}

\newcommand{\PEC}{\mathrm{PEC}}
\newcommand{\CDR}{\mathrm{CDR}}

\begin{document}

\title{Finite-shot operating windows for probabilistic error cancellation and Clifford data regression}

\author{Vicenzo Scavino}
\email{u201919346@upc.edu.pe}
\homepage[ORCID: ]{https://orcid.org/0009-0000-2472-9785}
\affiliation{Independent Researcher, Lima, Peru}

\date{June 19, 2026}

\begin{abstract}
Quantum error mitigation on noisy devices is limited not only by residual bias
but also by the shot noise and calibration errors introduced by the mitigation
procedure itself.  We derive finite-shot mean-square-error boundaries for
probabilistic error cancellation (PEC), Clifford data regression (CDR), and no
mitigation for noisy Pauli-observable estimates.  Exact PEC removes the target
bias under an exact noise inverse at the price of a quasi-probability variance
overhead, whereas population linear CDR can have smaller target-shot variance
but retains a calibration floor when the training and target noise responses do
not match.  This competition yields a finite CDR-dominant operating window whose
upper endpoint scales as $B_{\PEC=\CDR}(p)\propto 1/(\delta_1^2p)$, where
$\delta_1$ is the first-order CDR calibration mismatch.  We further prove a
target-response projection theorem showing that response-blind affine CDR
removes the first-order bias only when the target noise response is affine in
the ideal target value; otherwise a nonzero projection error gives an
irreducible local calibration floor.  The same mean-square-error formulation
extends to second-order calibration, commuting Pauli Hamiltonians, finite CDR
training shots, and residual PEC model bias.  A closed-form two-qubit
calculation and QAOA simulations support the predicted no-mitigation,
CDR-dominant, and PEC-dominant regimes.
\end{abstract}


\maketitle

\section{Introduction}
\label{sec:introduction}

Quantum error mitigation (QEM) aims to estimate ideal expectation values from
noisy quantum circuits without the full overhead of fault-tolerant error
correction.  Probabilistic error cancellation (PEC) uses a quasi-probability
inverse of a noise channel and can be unbiased when the noise model is exact
\cite{Temme2017,Endo2018,Takagi2021,Piveteau2022,vandenBerg2023}.  Clifford data
regression (CDR) instead learns a classical correction map from classically
simulable training circuits \cite{Czarnik2021,Strikis2021,Lowe2021}.  Both
families can reduce bias, with finite-shot performance governed by variance
overhead, shot allocation, and model mismatch
\cite{Cai2023,Bultrini2023,Russo2023,Wang2024}.

We focus on the finite-shot comparison that follows from this bias--variance
tension:
\begin{quote}
\emph{For a given noise rate $p$ and shot budget $B$, should one use PEC, CDR, or
no mitigation?}
\end{quote}
This question has a sharp PEC-CDR tension.  PEC removes the noisy bias under an
exact noise inverse, but its variance is inflated by the quasi-probability
one-norm \cite{Temme2017,Endo2018,Takagi2021,Piveteau2022,vandenBerg2023}.
CDR can have a smaller target-shot variance overhead, while calibration mismatch
produces a misspecification bias that persists under increasing target shots
\cite{Czarnik2021,Strikis2021,Lowe2021,Bultrini2023}.

The main result is that this trade-off creates a finite region in which CDR is
MSE-optimal.  At small noise, and under the local assumptions stated below, CDR
occupies the dominant region between a lower help-harm threshold and an upper
PEC-CDR crossover,
\begin{equation}
    B_{\CDR}^*(p)<B<B_{\PEC=\CDR}(p),
    \qquad
    B_{\PEC=\CDR}(p)\sim
    \frac{\kappa_{\PEC}-\kappa_{\CDR}}{\delta_1^2}\frac{1}{p}.
\end{equation}
Here $\delta_1$ is the first-order CDR calibration mismatch.  Informally, it is
the first-order slope error made when the noise response learned from the CDR
training family is extrapolated to the target circuit.  The window grows as
calibration improves and collapses quadratically as mismatch increases.

\paragraph{Scope.}
The analysis is local in noise strength and specific to the stated estimator,
observable, calibration, and shot-budget assumptions.  It complements global
QEM cost and lower-bound results
\cite{Takagi2022,Takagi2023,Tsubouchi2023,Quek2024} by giving explicit
finite-shot PEC/CDR comparison boundaries for analytically tractable settings.
The CDR statements are relative to a fixed training distribution or to a
specified target-conditioned training design; training sets that satisfy the
calibration condition can remove the first-order mismatch.

\paragraph{Contributions.}
First, we derive a finite-shot MSE comparison for exact PEC, linear CDR, and no
mitigation, including an exact PEC help-harm boundary for Pauli observables and
the PEC-CDR crossover law controlled by $\delta_1$.

Second, we identify the CDR calibration mechanism that controls this crossover.
The target-response projection theorem gives a sharp condition for when
response-blind affine CDR can remove the first-order target bias and when a
projection floor remains.

Third, we test the resulting boundaries in a closed-form two-qubit example and
in QAOA simulations.  The supplementary analyses examine Hamiltonian
observables, graph instances, Pauli-channel variants, angle changes, near-zero
$\mu_0$, circuit depth, finite CDR training shots, and device-derived noise
models.

\subsection{Related work}
\label{subsec:related_work}

PEC and zero-noise extrapolation were introduced as early QEM strategies for
short-depth circuits \cite{Temme2017,Li2017,Endo2018}.  Later work developed
practical demonstrations, digital noise-scaling frameworks, Pauli-Lindblad
implementations, hybrid PEC--ZNE variants, quasi-probability decompositions
with reduced overhead, and resource-theoretic sampling analyses
\cite{Song2019,Kandala2019,GiurgicaTiron2020,Mari2021,Takagi2021,Piveteau2022,vandenBerg2023}.
Recent work further studies PEC sampling overheads and faster PEC constructions
\cite{Scheiber2025,Chen2025}, as well as non-Clifford gate settings
\cite{Layden2026}.  Against this PEC-construction
literature, the present analysis fixes the estimator class and asks when PEC
outperforms its sampling overhead.  It derives an explicit finite-shot
help-harm threshold for Pauli observables and uses that threshold as one side of
the PEC-CDR comparison boundary.

CDR and related learning-based methods train correction maps from simulable or
near-Clifford data
\cite{Czarnik2021,Strikis2021,Lowe2021,Czarnik2025,PerezGuijarro2024,Liao2024,Liao2025}.
Those works define and extend regression-based mitigation ansatzes.  The present
analysis takes a complementary population-level view: for a fixed linear CDR
map, it characterizes first-order target calibration and translates the resulting
calibration mismatch into a closed-form PEC-CDR crossover.  The
target-response projection theorem further identifies the irreducible
response-blind affine mismatch over a target ensemble.

The closest numerical and benchmarking studies show that QEM performance is
operating-point dependent: it can vary with shot budget, noise model, circuit
family, observable, and calibration quality
\cite{Qin2023,Bultrini2023,LaRose2022,Cirstoiu2023,Russo2023,Kim2023,Wang2024,Govia2025,Niroula2025}.
Recent uncertainty-quantification work studies finite-shot propagation and PEC
resource overheads \cite{Prodius2026}, while mitigation-aware benchmarking work
accounts for finite-shot effects in optimization benchmarks \cite{Demarty2026}.
Shot-estimation analyses for noisy circuits quantify the sampling variance and
shot requirements of expectation-value estimates
\cite{SeksariaPrabhakar2025}.  Together, this
literature supports operating-point-aware comparison of mitigation methods.
The distinction here is analytic: prior benchmarking work maps finite-shot
regimes empirically, while this paper derives closed-form finite-shot
indifference boundaries between PEC, CDR, and no mitigation.  In particular, it
makes the first-order CDR calibration mismatch $\delta_1$ the explicit quantity
controlling the upper endpoint of the CDR-dominant window.

Finally, the lower-bound statement for response-blind affine CDR follows
standard statistical decision-theory usage \cite{Berger1985}.  The theorem is a
sharp statement about the specified affine response-blind class; nonlinear,
adaptive, or target-informed mitigation designs lie outside its hypothesis
class.

\paragraph{Organization.}
Section~\ref{sec:framework} introduces the common MSE formulation.
Section~\ref{sec:pec_boundary} derives the exact PEC finite-shot boundary.
Section~\ref{sec:cdr_calibration} analyzes the CDR calibration floor and
summarizes the target-response projection mechanism, with full details in
Appendix~\ref{app:cdr_mathematical_details}.  Section~\ref{sec:phase_diagram} combines
PEC and CDR into finite-shot operating boundaries.  The remaining main-text
sections extend the analysis to commuting Hamiltonians, give the two-qubit
closed-form example, and report the QAOA numerical validation.

\section{Common MSE formulation}
\label{sec:framework}

Let $U$ prepare $\rho_0=U\ket{0}\bra{0}U^\dagger$ and let $O$ be the measured
observable.  The ideal and noisy expectation values are
\begin{equation}
    \mu_0=\Tr[O\rho_0],
    \qquad
    \mu_p=\Tr[O\mathcal{N}_p(\rho_0)],
\end{equation}
and the noisy bias is
\begin{equation}
    b_N(p)=\mu_p-\mu_0.
\end{equation}
With $B$ target shots,
\begin{equation}
    \MSE_N(p,B)=b_N^2(p)+\frac{\sigma_N^2(p)}{B}.
    \label{eq:mse_noisy}
\end{equation}
For a Pauli observable measured projectively, $\sigma_N^2(p)=1-\mu_p^2$.

For any mitigation method $M$ with
\begin{equation}
    \MSE_M(p,B)=b_M^2(p)+\frac{\sigma_M^2(p)}{B},
\end{equation}
define
\begin{equation}
    D_M(p)=b_N^2(p)-b_M^2(p),
    \qquad
    A_M(p)=\sigma_M^2(p)-\sigma_N^2(p).
\end{equation}
The finite-shot gain over the unmitigated estimator is
\begin{equation}
    \Delta_M(p,B):=\MSE_N(p,B)-\MSE_M(p,B)
    =D_M(p)-\frac{A_M(p)}{B}.
    \label{eq:Delta_general}
\end{equation}
Thus, in the common threshold case $D_M(p)>0$ and $A_M(p)>0$, mitigation helps
if and only if
\begin{equation}
    B> B_M^*(p):=\frac{A_M(p)}{D_M(p)}.
\end{equation}
The full sign table, including degenerate cases, is given in
Appendix~\ref{app:help_harm_cases}.

\begin{theorem}[Generic biased--unbiased operating window]
\label{thm:generic_bias_variance_window}
Consider two mitigation methods for the same target observable.  Method $U$ is
unbiased on the considered local range and has
\begin{equation}
    b_U(p)=0,\qquad \sigma_U^2(p)-\sigma_L^2(p)=\kappa p+O(p^2),
    \qquad \kappa>0,
\end{equation}
while method $L$ has lower leading variance but a local residual bias
\begin{equation}
    b_L(p)=\delta p+O(p^2),\qquad \delta\ne0.
\end{equation}
Then the indifference budget between the two methods satisfies
\begin{equation}
    B_{U=L}(p)
    =\frac{\sigma_U^2(p)-\sigma_L^2(p)}{b_L^2(p)-b_U^2(p)}
    =\frac{\kappa}{\delta^2}\frac{1}{p}+O(1).
    \label{eq:generic_unbiased_biased_crossover}
\end{equation}
Consequently, a biased lower-variance method can beat the unbiased method only
below a finite crossover budget, and this crossover collapses quadratically as
the first-order residual bias grows.  For sufficiently small $p>0$, this
positive crossover is unique and method $L$ has smaller MSE exactly for
$B<B_{U=L}(p)$.
\end{theorem}

\begin{proof}
The difference of the two MSEs is
\begin{equation}
    \MSE_U(p,B)-\MSE_L(p,B)
    =-b_L^2(p)+\frac{\sigma_U^2(p)-\sigma_L^2(p)}{B}.
\end{equation}
Equating this expression to zero and substituting the local expansions gives
\eqref{eq:generic_unbiased_biased_crossover}.  For sufficiently small $p>0$,
both $\sigma_U^2(p)-\sigma_L^2(p)$ and $b_L^2(p)$ are positive, and the
difference is strictly decreasing as a function of $B$; hence the crossing is
unique and its sign gives the stated below-crossover dominance of $L$.
\end{proof}

This theorem is the method-agnostic bias--variance mechanism used throughout
the paper.  Exact PEC and population CDR instantiate it with $U=\PEC$ and
$L=\CDR$: PEC removes the target bias but pays a quasi-probability variance
overhead, whereas CDR can have smaller target-shot variance while retaining a
calibration-induced residual bias.

\section{Exact PEC finite-shot boundary}
\label{sec:pec_boundary}

For a $q$-qubit gate, $q\in\{1,2\}$, use the local depolarizing convention
\begin{equation}
    \mathcal{D}_{p_q}^{(q)}(\rho)
    =(1-p_q)\rho+\frac{p_q}{4^q-1}\sum_{P\in\mathcal{P}_q\setminus\{I\}}
    P\rho P^\dagger,
    \label{eq:depol_channel}
\end{equation}
with $p_1=p$ and $p_2=2p$ unless otherwise stated.

\begin{assumption}[PEC domain and sampling model]
\label{ass:pec_domain}
The measured observable is a Pauli observable with outcomes $\pm1$; the local
depolarizing model \eqref{eq:depol_channel} is exact; quasi-probability samples
are independent; and
\begin{equation}
    0\le p<\frac{15}{32}.
\end{equation}
This keeps both local Pauli eigenvalues positive under $p_1=p$ and $p_2=2p$.
\end{assumption}

For every non-identity $q$-qubit Pauli $Q\ne I$,
\begin{equation}
    \mathcal{D}_{p_q}^{(q)}(Q)=\lambda_q(p_q)Q,
    \qquad
    \lambda_q(p_q)=1-\frac{4^q}{4^q-1}p_q.
    \label{eq:lambda_q}
\end{equation}
The Pauli-symmetric inverse decomposition has local one-norm
\begin{equation}
    \gamma_g^{(q)}(p_q)
    =\frac{1+\frac{4^q-2}{4^q-1}p_q}
    {1-\frac{4^q}{4^q-1}p_q}.
    \label{eq:gamma_q}
\end{equation}
To see this, write the inverse as a symmetric quasi-probability combination of
Pauli conjugations.  Its identity coefficient is
$(4^q-(1-\lambda_q(p_q)))/(\lambda_q(p_q)4^q)$ and each non-identity Pauli
coefficient is $-(1-\lambda_q(p_q))/(\lambda_q(p_q)4^q)$; summing the absolute
values gives
\eqref{eq:gamma_q}, in line with standard PEC decompositions
\cite{Temme2017,Takagi2021}.
Thus
\begin{equation}
    \gamma_g^{(1)}(p)=\frac{1+\frac{2}{3}p}{1-\frac{4}{3}p},
    \qquad
    \gamma_g^{(2)}(2p)=\frac{1+\frac{28}{15}p}{1-\frac{32}{15}p}.
\end{equation}
For a circuit with $N_{1q}$ one-qubit noisy locations and $N_{2q}$ two-qubit
noisy locations,
\begin{equation}
    \Gamma(p)=\prod_{g\in 1q}\gamma_g^{(1)}(p)
    \prod_{g\in 2q}\gamma_g^{(2)}(2p),
    \qquad
    \ln\Gamma(p)=\left(2N_{1q}+4N_{2q}\right)p+O(p^2).
    \label{eq:Gamma_circuit}
\end{equation}

Under exact PEC, a single mitigated sample can be written as
\begin{equation}
    X_{\PEC}=\Gamma(p)SY,
\end{equation}
where $S\in\{-1,+1\}$ is the sampled quasi-probability sign and $Y$ is the
measurement outcome of the sampled implementable circuit.  Since $S^2=Y^2=1$,
\begin{equation}
    \sigma_{\PEC}^2(p)=\Var[X_{\PEC}]=\Gamma^2(p)-\mu_0^2.
    \label{eq:sigma_pec_exact}
\end{equation}

\begin{theorem}[Exact PEC finite-shot boundary for Pauli observables]
\label{prop:pec_exact}
Under Assumption~\ref{ass:pec_domain}, exact PEC has
\begin{equation}
    \MSE_{\PEC}(p,B)=\frac{\Gamma^2(p)-\mu_0^2}{B}.
\end{equation}
Its finite-shot gain over the noisy estimator is
\begin{equation}
    \Delta_{\PEC}(p,B)=b_N^2(p)-\frac{A_{\PEC}(p)}{B},
\end{equation}
where
\begin{equation}
    A_{\PEC}(p)=\Gamma^2(p)-1+\mu_p^2-\mu_0^2.
    \label{eq:Apec_pauli}
\end{equation}
If $A_{\PEC}(p)>0$ and $b_N(p)\ne0$, the PEC help-harm threshold is
\begin{equation}
    B_{\PEC}^*(p)
    =\frac{\Gamma^2(p)-1+\mu_p^2-\mu_0^2}{(\mu_p-\mu_0)^2}.
    \label{eq:Bpec_exact}
\end{equation}
\end{theorem}

\begin{proof}
Exact PEC is unbiased, so $b_{\PEC}=0$ and $D_{\PEC}=b_N^2$.  For Pauli
measurements, $\sigma_N^2(p)=1-\mu_p^2$ and
\eqref{eq:sigma_pec_exact} gives
$\sigma_{\PEC}^2(p)=\Gamma^2(p)-\mu_0^2$.  Therefore
$A_{\PEC}=\sigma_{\PEC}^2-\sigma_N^2$ equals \eqref{eq:Apec_pauli}, and the
threshold follows from \eqref{eq:Delta_general}.
\end{proof}

\begin{corollary}[Abstract PEC one-norm form and Pauli-channel robustness]
\label{cor:abstract_pec}
The boundary \eqref{eq:Bpec_exact} does not rely on the depolarizing form except
through the exact inverse one-norm.  Consider any exact quasi-probability inverse
for a Pauli-channel, or for a circuit-level Pauli-Lindblad model diagonal in the
Pauli basis, with i.i.d. samples whose single-shot estimator for a Pauli
observable has the form
\begin{equation}
    X_{\PEC}=\Gamma_{\mathcal Q}(p)S Y,
    \qquad S,Y\in\{-1,+1\},
\end{equation}
and satisfies $\E[X_{\PEC}]=\mu_0$.  Then
\begin{equation}
    \MSE_{\PEC}(p,B)
    =\frac{\Gamma_{\mathcal Q}^2(p)-\mu_0^2}{B},
    \qquad
    B_{\PEC}^*(p)
    =\frac{\Gamma_{\mathcal Q}^2(p)-1+\mu_p^2-\mu_0^2}
    {(\mu_p-\mu_0)^2},
\end{equation}
whenever the denominator is nonzero and the numerator is positive.  If
\begin{equation}
    \Gamma_{\mathcal Q}^2(p)=1+\kappa_\Gamma p+O(p^2),
    \qquad
    \mu_p=\mu_0+\beta_Tp+O(p^2),
    \qquad \beta_T\ne0,
\end{equation}
then
\begin{equation}
    B_{\PEC}^*(p)
    =
    \frac{\kappa_\Gamma+2\mu_0\beta_T}{\beta_T^2}\frac{1}{p}+O(1).
    \label{eq:Bpec_abstract_small_p}
\end{equation}
\end{corollary}

\begin{proof}
The proof of Theorem~\ref{prop:pec_exact} uses only exact unbiasedness and
$S^2=Y^2=1$.  Therefore the same variance calculation holds with
$\Gamma_{\mathcal Q}$ in place of the local-depolarizing $\Gamma$.  The small-$p$
form follows by expanding
$\Gamma_{\mathcal Q}^2(p)-1+\mu_p^2-\mu_0^2$.
\end{proof}

Corollary~\ref{cor:abstract_pec} is the form used by the finite-shot comparison
theorem below.  The local depolarizing convention gives a closed-form
$\Gamma(p)$, while general Pauli-channel or sparse Pauli-Lindblad models enter
through the one-norm of the chosen exact inverse decomposition
\cite{vandenBerg2023}.

If $\mu_p=\mu_0+\beta_Tp+O(p^2)$ with $\beta_T\ne0$, then
\begin{equation}
    B_{\PEC}^*(p)
    =\frac{2(2N_{1q}+4N_{2q})+2\mu_0\beta_T}{\beta_T^2}\frac{1}{p}+O(1),
\end{equation}
whenever the numerator is positive.  Thus exact PEC begins helping only when the
shot budget is large enough to pay for the quasi-probability variance overhead.

\begin{remark}[Leading positivity in the Clifford-Pauli case]
For a Clifford circuit and Pauli observable under the local depolarizing
convention, the noisy expectation has the multiplicative form
$\mu_p=\Lambda(p)\mu_0$ with
$\Lambda(p)=1-R_O p+O(p^2)$ and
$0\le R_O\le \frac{4}{3}N_{1q}+\frac{32}{15}N_{2q}$.  Hence the leading
coefficient of $A_{\PEC}(p)$ satisfies
\begin{equation}
    2(2N_{1q}+4N_{2q})-2R_O\mu_0^2
    \ge \frac{4}{3}N_{1q}+\frac{56}{15}N_{2q},
\end{equation}
using $\mu_0^2\le1$.  Thus the PEC variance-overhead numerator is locally
positive in this common Clifford-Pauli setting whenever the circuit contains a
noisy location.  For general non-Clifford targets, the conditional statement
above is the safer form and no such uniform lower bound is assumed.
\end{remark}

\section{CDR calibration floor}
\label{sec:cdr_calibration}

Let $\{C_j\}_{j=1}^{N_t}$ be near-Clifford training circuits with noisy and ideal
expectations
\begin{equation}
    x_j(p)=\mu_p^{(j)},
    \qquad
    y_j=\mu_0^{(j)}.
\end{equation}
Linear CDR learns $y\approx ax+b$.  In the pretrained population setting,
\begin{equation}
    (a^*(p),b^*(p))
    =\argmin_{a,b}\E_{\mathcal{T}}[(y-a x(p)-b)^2],
\end{equation}
and the target residual bias is
\begin{equation}
    b_{\mathrm{mis}}(p)=a^*(p)\mu_p+b^*(p)-\mu_0.
    \label{eq:bmis_def}
\end{equation}
This is a model misspecification bias: it remains even with infinitely many
target shots and population training data under target-response mismatch between
the affine model learned from the training distribution and the target.

\begin{assumption}[Analytic CDR response]
\label{ass:cdr_response}
For each training circuit and for the target circuit,
\begin{align}
    \mu_p^{(j)}&=\mu_0^{(j)}+\beta_jp+\chi_jp^2+O(p^3),\label{eq:train_expansion}\\
    \mu_p&=\mu_0+\beta_Tp+\chi_Tp^2+O(p^3).\label{eq:target_expansion}
\end{align}
The training-circuit remainders are uniform over the support of
$\mathcal{T}$, or are dominated so that expectations, variances, and
covariances may be expanded termwise.  This condition is automatic for a
finite training set.
\end{assumption}

The least-squares algebra is collected in
Appendix~\ref{app:cdr_mathematical_details}.  Its main consequence for the
operating-window argument is that the population CDR map has local expansion
\begin{equation}
    a^*(p)=1-a_1p+O(p^2),
    \qquad
    b^*(p)=-b_1p+O(p^2),
\end{equation}
and hence the target residual has the form
\begin{equation}
    b_{\mathrm{mis}}(p)=\delta_1p+\delta_2p^2+O(p^3),
    \qquad
    \delta_1=\beta_T-a_1\mu_0-b_1.
\end{equation}
The coefficient $\delta_1$ is the first-order CDR calibration mismatch: it is
zero exactly when the target's first-order noise response lies on the affine
response learned from the training family.  Thus a nonzero $\delta_1$ gives a
population calibration floor even with infinite target shots and exact population
training data.

The appendix also records the higher-order and statistical consequences used in
the numerical tests: when $\delta_1=0$ but $\delta_2\ne0$, the PEC-CDR
crossover scales as $p^{-3}$; response-blind affine training has a projection
floor unless the target response is affine in the ideal value; and finite CDR
training adds a coefficient-estimation term controlled by the training Gram
conditioning.  These results are not needed to derive the main PEC-CDR window,
but they explain the calibration-hierarchy, projection-floor, and finite-training
tests reported in Appendix~\ref{app:phase8_plus_experimental_layers}.

In the pretrained setting,
\begin{equation}
    \MSE_{\CDR}(p,B)=b_{\mathrm{mis}}^2(p)
    +a^{*2}(p)\frac{\sigma_N^2(p)}{B}.
\end{equation}
Define
\begin{equation}
    D_{\CDR}(p)=b_N^2(p)-b_{\mathrm{mis}}^2(p),
    \qquad
    A_{\CDR}(p)=(a^{*2}(p)-1)\sigma_N^2(p).
    \label{eq:Acdr_Dcdr}
\end{equation}
For small $p$, the CDR expansion in
Appendix~\ref{app:cdr_mathematical_details} gives
$A_{\CDR}(p)=-2a_1\bigl(1-\mu_0^2\bigr)p+O(p^2)$, so the leading CDR
variance-overhead coefficient used in the operating-window theorem below is
$\kappa_{\CDR}=-2a_1(1-\mu_0^2)$.
When $D_{\CDR}(p)>0$ and $A_{\CDR}(p)>0$, the CDR help-harm threshold is
\begin{equation}
    B_{\CDR}^*(p)
    =\frac{(a^{*2}(p)-1)\sigma_N^2(p)}
    {b_N^2(p)-b_{\mathrm{mis}}^2(p)}.
    \label{eq:Bcdr}
\end{equation}
If $D_{\CDR}(p)\le0$ and $A_{\CDR}(p)\ge0$, the MSE ordering favors the
unmitigated estimator at every target-shot budget, strictly so unless
$D_{\CDR}=A_{\CDR}=0$.  Finite training adds coefficient-estimation variance to
CDR, but it leaves the calibration-induced bias floor unchanged.

More generally, if
\begin{equation}
    b_N(p)=\alpha_rp^r+O(p^{r+1}),
    \qquad
    b_{\mathrm{mis}}(p)=\delta_sp^s+O(p^{s+1}),
\end{equation}
then CDR improves asymptotic bias only while
$|b_{\mathrm{mis}}(p)|<|b_N(p)|$.  When $s>r$, the leading calibration ceiling is
\begin{equation}
    p_{\CDR}^*=\left(\frac{|\alpha_r|}{|\delta_s|}\right)^{1/(s-r)}.
\end{equation}
When $s=r$, local improvement requires $|\delta_s|<|\alpha_r|$; when $s<r$, the
CDR residual dominates the noisy bias locally.

\section{PEC-CDR operating window}
\label{sec:phase_diagram}

The finite-shot comparison rule is based on the gains over the noisy estimator:
\begin{align}
    \Delta_{\PEC}(p,B)&=b_N^2(p)-\frac{A_{\PEC}(p)}{B},\label{eq:Delta_pec}\\
    \Delta_{\CDR}(p,B)&=D_{\CDR}(p)-\frac{A_{\CDR}(p)}{B}.\label{eq:Delta_cdr}
\end{align}
No mitigation has gain zero.  Therefore the MSE-optimal choice is the method with
the largest value among $0$, $\Delta_{\PEC}$, and $\Delta_{\CDR}$.

\begin{definition}[Finite-shot PEC-CDR operating regions]
\label{prop:phase_partition}
Assume exact PEC and pretrained CDR.  The positive-gain operating regions are
\begin{align}
    \mathcal{R}_{\PEC}&=\{(p,B):\Delta_{\PEC}>0,
    \Delta_{\PEC}>\Delta_{\CDR}\},\\
    \mathcal{R}_{\CDR}&=\{(p,B):\Delta_{\CDR}>0,
    \Delta_{\CDR}>\Delta_{\PEC}\},\\
    \mathcal{R}_{0}&=\{(p,B):\Delta_{\PEC}\le0,
    \Delta_{\CDR}\le0\}.
\end{align}
Within $\mathcal{R}_{\PEC}$, PEC gives the smallest MSE; within
$\mathcal{R}_{\CDR}$, CDR gives the smallest MSE; within $\mathcal{R}_0$, no
mitigation should be used.  Equalities are indifference boundaries on which the
MSE-optimal choice is set-valued.
\end{definition}

The PEC-CDR indifference curve is obtained from
$\Delta_{\PEC}(p,B)=\Delta_{\CDR}(p,B)$.  Since
$b_N^2-D_{\CDR}=b_{\mathrm{mis}}^2$, the crossover budget is
\begin{equation}
    B_{\PEC=\CDR}(p)=
    \frac{A_{\PEC}(p)-A_{\CDR}(p)}{b_{\mathrm{mis}}^2(p)},
    \label{eq:pec_cdr_crossover_budget}
\end{equation}
whenever the right-hand side is positive.  If $A_{\PEC}\le A_{\CDR}$ and
$b_{\mathrm{mis}}\ne0$, the gain gap favors exact PEC throughout that local
slice in $p$.  If $A_{\PEC}>A_{\CDR}$, CDR can dominate below the crossover
budget when both methods have positive gain, while exact PEC dominates above it.

\begin{theorem}[CDR operating-window theorem]
\label{thm:cdr_operating_window}
Assume, as $p\to0$,
\begin{align}
    b_N(p)&=\alpha p+O(p^2),\qquad \alpha\ne0,\\
    b_{\mathrm{mis}}(p)&=\delta p+O(p^2),\\
    A_{\PEC}(p)&=\kappa_{\PEC}p+O(p^2),\\
    A_{\CDR}(p)&=\kappa_{\CDR}p+O(p^2),
\end{align}
with $\kappa_{\PEC}>\kappa_{\CDR}>0$ and $\alpha^2>\delta^2>0$.  In the
single-Pauli CDR setting summarized in Section~\ref{sec:cdr_calibration},
$\delta=\delta_1$;
the abstract symbols also cover the Hamiltonian substitutions of
Section~\ref{sec:hamiltonian_extension}.  For every
sufficiently small fixed $p>0$, any CDR-dominant positive-gain interval is
bounded above by
\begin{equation}
    B_{\PEC=\CDR}(p)=
    \frac{\kappa_{\PEC}-\kappa_{\CDR}}{\delta^2}\frac{1}{p}+O(1).
    \label{eq:cdr_window_upper}
\end{equation}
When CDR has a nonempty local dominance window, its leading form is
\begin{equation}
    \mathcal{W}_{\CDR}(p)=
    \left(
    \frac{\kappa_{\CDR}}{\alpha^2-\delta^2}\frac{1}{p},
    \frac{\kappa_{\PEC}-\kappa_{\CDR}}{\delta^2}\frac{1}{p}
    \right)+O(1).
    \label{eq:cdr_window_interval}
\end{equation}
At leading order this window is nonempty exactly when
\begin{equation}
    \delta^2<\alpha^2\frac{\kappa_{\PEC}-\kappa_{\CDR}}{\kappa_{\PEC}}.
    \label{eq:cdr_window_nonempty}
\end{equation}
Equivalently,
\begin{equation}
    \frac{B_{\PEC=\CDR}(p)}{B_{\PEC}^*(p)}
    \to
    \frac{\alpha^2(\kappa_{\PEC}-\kappa_{\CDR})}
    {\delta^2\kappa_{\PEC}}.
    \label{eq:window_ratio}
\end{equation}
Thus the budget scale over which CDR can dominate exact PEC grows like
$1/\delta^2$ as calibration mismatch vanishes and collapses quadratically as the
first-order mismatch increases.
\end{theorem}

\begin{proof}
The CDR help-harm threshold follows from \eqref{eq:Bcdr}:
\begin{equation}
    B_{\CDR}^*(p)=\frac{\kappa_{\CDR}}{\alpha^2-\delta^2}\frac{1}{p}+O(1).
\end{equation}
The PEC-CDR crossover follows from \eqref{eq:pec_cdr_crossover_budget} because
$A_{\PEC}-A_{\CDR}=(\kappa_{\PEC}-\kappa_{\CDR})p+O(p^2)$ and
$b_{\mathrm{mis}}^2=\delta^2p^2+O(p^3)$.  Therefore the CDR-dominant operating
window lies between these endpoints.  Comparing the leading lower and
upper endpoints gives \eqref{eq:cdr_window_nonempty}; dividing by
$B_{\PEC}^*(p)=\kappa_{\PEC}\alpha^{-2}p^{-1}+O(1)$ gives
\eqref{eq:window_ratio}.
\end{proof}

\begin{remark}[Window displacement near $\mu_0=0$]
For near-multiplicative training and target families ($b_1\approx0$,
$\mu_p\approx\Lambda_T(p)\mu_0$), both $\alpha$ and $\delta$ are proportional
to $\mu_0$, while $\kappa_{\PEC}$ and $\kappa_{\CDR}$ stay bounded away from
zero.  Both endpoints of \eqref{eq:cdr_window_interval} therefore scale like
$1/\mu_0^2$: the leading-order nonemptiness condition
\eqref{eq:cdr_window_nonempty} is unchanged, but the entire window is pushed
beyond any fixed shot-budget range as $\mu_0\to0$.  The near-zero-$\mu_0$
robustness test in Appendix~\ref{sec:phase7_supplement} observes exactly this
displacement on the selected budget grid.
\end{remark}

\begin{remark}[Residual PEC model bias]
If the assumed-model PEC estimator has residual first-order bias
$b_{\PEC,\mathrm{res}}(p)=\epsilon p+O(p^2)$, the same MSE equality gives the
local deformation
\begin{equation}
    B_{\PEC=\CDR}^{\mathrm{viol}}(p)
    =
    \frac{\kappa_{\PEC}-\kappa_{\CDR}}{\delta^2-\epsilon^2}\frac{1}{p}
    +O(1),
\end{equation}
whenever $\delta^2>\epsilon^2$.  If $\epsilon^2\ge\delta^2$, PEC no longer has
the smaller leading bias, so the local point falls outside the exact-PEC
ordering premise of Theorem~\ref{thm:cdr_operating_window}.  This is the local
operating-boundary effect of model violation \cite{Govia2025}.
\end{remark}

Theorem~\ref{thm:cdr_operating_window} is the main message of the paper.  The
PEC-CDR comparison is an operating-point statement: the ordering is controlled
by the finite-shot budget and by the first-order calibration mismatch of the CDR
training family.

\section{Commuting Hamiltonian observables}
\label{sec:hamiltonian_extension}

The single-Pauli formulas above are the cleanest setting for the PEC boundary,
but variational applications often estimate a Hamiltonian
\begin{equation}
    H=\sum_{e=1}^{m} w_e P_e,
    \label{eq:commuting_hamiltonian}
\end{equation}
where the $P_e$ are commuting Pauli observables.  MaxCut is the standard case:
$P_e=Z_iZ_j$ for graph edges $e=(i,j)$, and all edge terms can be read out from
the same computational-basis shot.

Let $Y(p)=(Y_1(p),\ldots,Y_m(p))^\top\in\{\pm1\}^m$ be the vector of Pauli
outcomes from one noisy shot.  Define
\begin{equation}
    \boldsymbol{\mu}_p=\E[Y(p)],
    \qquad
    C_p=\E[Y(p)Y(p)^\top],
    \qquad
    \Sigma_p=C_p-\boldsymbol{\mu}_p\boldsymbol{\mu}_p^\top .
\end{equation}
The noisy scalar estimator for $H$ has
\begin{equation}
    \mu_{p,H}=w^\top\boldsymbol{\mu}_p,
    \qquad
    \sigma_{N,H}^2(p)=w^\top\Sigma_p w,
    \qquad
    b_{N,H}(p)=w^\top(\boldsymbol{\mu}_p-\boldsymbol{\mu}_0).
    \label{eq:noisy_hamiltonian_mse}
\end{equation}
Thus simultaneous measurement changes the variance term from a sum of
independent variances to a covariance quadratic form.

For exact PEC, a quasi-probability sample produces a signed implementable
circuit and a simultaneous outcome vector $Y_{\mathcal Q}\in\{\pm1\}^m$.  With
one-norm $\Gamma_{\mathcal Q}(p)$, the Hamiltonian sample is
\begin{equation}
    X_{\PEC,H}=\Gamma_{\mathcal Q}(p)S\,w^\top Y_{\mathcal Q}.
\end{equation}
Exactness gives $\E[X_{\PEC,H}]=\mu_{0,H}=w^\top\boldsymbol{\mu}_0$, while
\begin{equation}
    \sigma_{\PEC,H}^2(p)
    =
    \Gamma_{\mathcal Q}^2(p)\,w^\top C_{\mathcal Q}(p)w-\mu_{0,H}^2,
    \qquad
    C_{\mathcal Q}(p)=\E_{\mathcal Q}[Y_{\mathcal Q}Y_{\mathcal Q}^\top].
    \label{eq:pec_hamiltonian_variance}
\end{equation}
In special sampling conventions where the absolute quasi-probability ensemble
has the same second moments as the noisy target ensemble for the measured
commuting set, $C_{\mathcal Q}(p)=C_p$ and
\eqref{eq:pec_hamiltonian_variance} reduces to
\begin{equation}
    \Gamma_{\mathcal Q}^2(p)
    \sum_{e,e'}w_ew_{e'}\langle P_eP_{e'}\rangle_p-\mu_{0,H}^2.
\end{equation}
The general expression \eqref{eq:pec_hamiltonian_variance} is the safer object:
it records the second moments of the actual PEC sampling ensemble.

\begin{corollary}[Hamiltonian finite-shot boundary]
\label{cor:hamiltonian_boundary}
For a commuting Pauli Hamiltonian measured as above, exact PEC has
\begin{equation}
    \MSE_{\PEC,H}(p,B)=\frac{\sigma_{\PEC,H}^2(p)}{B},
\end{equation}
and its gain over the noisy Hamiltonian estimator is
\begin{equation}
    \Delta_{\PEC,H}(p,B)
    =b_{N,H}^2(p)-\frac{A_{\PEC,H}(p)}{B},
    \qquad
    A_{\PEC,H}(p)=\sigma_{\PEC,H}^2(p)-w^\top\Sigma_p w.
\end{equation}
When $A_{\PEC,H}(p)>0$ and $b_{N,H}(p)\ne0$,
\begin{equation}
    B_{\PEC,H}^*(p)=\frac{A_{\PEC,H}(p)}{b_{N,H}^2(p)}.
\end{equation}
\end{corollary}

For CDR, one may train a scalar correction directly on Hamiltonian values, or
train edgewise corrections and aggregate them.  In the scalar case the previous
CDR formulas hold with $\mu_0,\mu_p$ replaced by $\mu_{0,H},\mu_{p,H}$ and the
Pauli shot variance $\sigma_N^2(p)=1-\mu_p^2$ replaced by
$\sigma_{N,H}^2(p)=w^\top\Sigma_p w$.  In an
edgewise first-order model,
\begin{equation}
    \boldsymbol{\mu}_p
    =\boldsymbol{\mu}_0+\boldsymbol{\beta}_T p+O(p^2),
    \qquad
    \boldsymbol{b}_{\mathrm{mis}}(p)=\boldsymbol{\delta}\,p+O(p^2),
\end{equation}
the aggregate Hamiltonian responses are
\begin{align}
    b_{N,H}(p)&=\alpha_Hp+O(p^2),
    &\alpha_H&=w^\top\boldsymbol{\beta}_T,\\
    b_{\mathrm{mis},H}(p)&=\delta_Hp+O(p^2),
    &\delta_H&=w^\top\boldsymbol{\delta}.
\end{align}
Per-edge first-order calibration, $\boldsymbol{\delta}=0$, is sufficient but not
necessary for the Hamiltonian to be first-order calibrated; cancellations can
make $w^\top\boldsymbol{\delta}=0$ for the aggregate observable.  With
$A_{\PEC,H}(p)=\kappa_{\PEC,H}p+O(p^2)$ and
$A_{\CDR,H}(p)=\kappa_{\CDR,H}p+O(p^2)$, the operating-window theorem applies
verbatim after the replacements
\begin{equation}
    \alpha\mapsto\alpha_H,\qquad
    \delta\mapsto\delta_H,\qquad
    \kappa_{\PEC}\mapsto\kappa_{\PEC,H},\qquad
    \kappa_{\CDR}\mapsto\kappa_{\CDR,H}.
\end{equation}
For reference, the leading CDR variance-overhead coefficients are
$\kappa_{\CDR,H}=-2a_{1,H}\,w^\top\Sigma_0w$ for the scalar-trained correction,
with $a_{1,H}$ the first-order slope coefficient of the scalar regression, and
$\kappa_{\CDR,H}=-2\,(w\circ a_1)^\top\Sigma_0\,w$ for the edgewise scheme,
where $(w\circ a_1)_e=w_ea_{1,e}$; both reduce to
$\kappa_{\CDR}=-2a_1(1-\mu_0^2)$ in the single-Pauli case.
This is the covariance-level extension needed for full MaxCut energy
validation.

\section{Closed-form two-qubit example}
\label{sec:two_qubit_example}

This section gives a small analytical example whose role is to display the
finite-shot operating regions without Monte Carlo simulation.  The example is
kept simple: the entangling layer does not change the ideal value of the chosen
observable, so the calculation isolates noise-response mismatch.

Consider
\begin{equation}
    \ket{00}
    \xrightarrow{R_y(\theta_1)\otimes R_y(\theta_2)}
    \xrightarrow{CZ}
    \text{measure } Z_1Z_2.
\end{equation}
Because $CZ$ commutes with $Z_1Z_2$,
\begin{equation}
    \mu_0(\theta)=\cos\theta_1\cos\theta_2.
\end{equation}
With local depolarizing noise after each one-qubit rotation and after the
entangling layer,
\begin{equation}
    \mu_{T,p}=\Lambda_T(p)\mu_0,
    \qquad
    \Lambda_T(p)=\lambda_1(p)^2\lambda_2(2p)
    =1-\frac{24}{5}p+O(p^2).
\end{equation}
For exact PEC, $\Gamma_T(p)=\gamma_1(p)^2\gamma_2(2p)$ and the PEC boundary is
\eqref{eq:Bpec_exact}.

Now train CDR on circuits that remove the entangling layer while preserving the
same ideal-value family:
\begin{equation}
    \ket{00}
    \xrightarrow{R_y(\theta_1)\otimes R_y(\theta_2)}
    \text{measure } Z_1Z_2.
\end{equation}
The training response is
\begin{equation}
    \Lambda_{\mathrm{tr}}(p)=\lambda_1(p)^2
    =1-\frac{8}{3}p+O(p^2).
\end{equation}
Since $x_{\mathrm{tr}}(p)=\Lambda_{\mathrm{tr}}(p)y$, the population regression
is exactly
\begin{equation}
    a^*(p)=\frac{1}{\Lambda_{\mathrm{tr}}(p)},
    \qquad
    b^*(p)=0.
\end{equation}
This holds for any training angle distribution with $V_y>0$, because the
relation $x_{\mathrm{tr}}=\Lambda_{\mathrm{tr}}y$ holds circuit by circuit.
Applied to the target,
\begin{equation}
    a^*(p)\mu_{T,p}-\mu_0
    =\left(\frac{\Lambda_T(p)}{\Lambda_{\mathrm{tr}}(p)}-1\right)\mu_0
    =\bigl(\lambda_2(2p)-1\bigr)\mu_0
    =-\frac{32}{15}\mu_0p,
\end{equation}
with no higher-order corrections, because the one-qubit factors cancel exactly
in the ratio $\Lambda_T/\Lambda_{\mathrm{tr}}=\lambda_2(2p)$.  Thus
$\delta_1=-\frac{32}{15}\mu_0$ and $\delta_2=0$: CDR has a purely first-order
residual even with a perfectly learned population regression on the training
family.

For $\theta_1=\pi/4$, $\theta_2=\pi/3$, and $\mu_0=\sqrt{2}/4$, direct evaluation
of the closed-form boundaries gives the operating regions in
Figure~\ref{fig:two_qubit_phase_diagram}.  CDR begins helping at a smaller shot
budget, but exact PEC wins beyond the PEC-CDR crossover because CDR retains a
first-order calibration bias.  For these angles the window-nonemptiness
condition \eqref{eq:cdr_window_nonempty} holds with room to spare:
$\delta^2/\alpha^2=16/81\approx0.20$, while
$(\kappa_{\PEC}-\kappa_{\CDR})/\kappa_{\PEC}=76/111\approx0.68$.  The numerical
threshold values used to generate the figure are archived with the reproduction
outputs.

\begin{figure}
\centering
\includegraphics[width=0.86\textwidth]{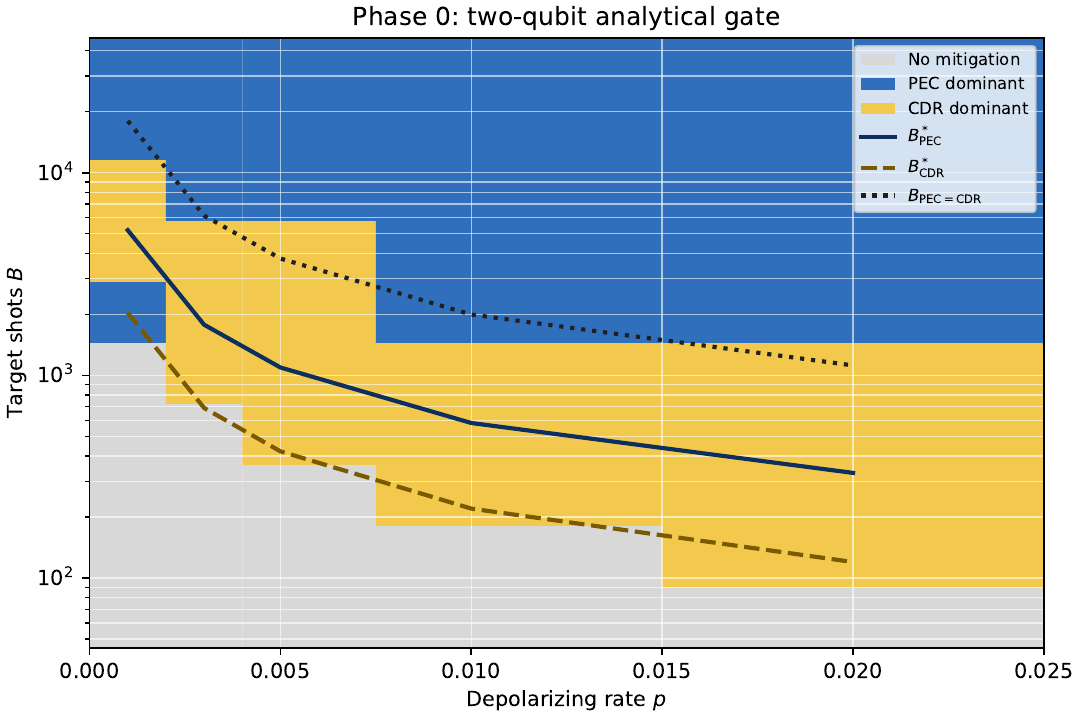}
\caption{Closed-form operating regions for the two-qubit example.  The gray
region is no mitigation, the blue region is PEC-dominant, and the yellow region
is CDR-dominant.  The solid, dashed, and dotted curves show
$B_{\PEC}^*(p)$, $B_{\CDR}^*(p)$, and $B_{\PEC=\CDR}(p)$, respectively.}
\label{fig:two_qubit_phase_diagram}
\end{figure}

\section{QAOA finite-shot validation}
\label{sec:qaoa_validation}

We test the finite-shot boundaries on a depth-one QAOA MaxCut circuit on a
four-qubit weighted graph with edges
\begin{equation}
    (0,1),(1,2),(2,3),(3,0),(0,2)
\end{equation}
and weights $1.0,0.8,1.2,0.9,0.7$.  The target angles are
$\gamma=0.72$ and $\beta=0.38$, with the depth-one ansatz
\begin{equation}
    \ket{\psi(\gamma,\beta)}
    =e^{-i\beta\sum_iX_i}\,e^{-i\gamma C}\,\ket{+}^{\otimes4},
    \qquad
    C=\sum_{(i,j)\in E}\frac{w_{ij}}{2}\bigl(1-Z_iZ_j\bigr).
\end{equation}
We first use the single edge observable $O_e=Z_0Z_1$ to isolate the scalar
boundary; the Hamiltonian covariance extension in
Section~\ref{sec:hamiltonian_extension} gives the corresponding bookkeeping for
the full MaxCut energy.  The ideal edge value is
\begin{equation}
    \mu_0=-0.40648296.
\end{equation}
The local depolarizing convention is the same as in
Section~\ref{sec:pec_boundary}.

Linear CDR is trained from $48$ restricted-subgraph QAOA circuits.  Each training
circuit contains the measured edge and, for one third of the training set, one
additional graph edge.  This creates a training family with a smaller two-qubit
noise response than the full target circuit while keeping the training data
correlated with the target observable.  For each noise rate, the
population linear CDR coefficients are fit from exact noisy and ideal training
expectations.  The finite-shot layer then uses $12{,}000$ Monte Carlo repetitions
per $(p,B)$ cell.  The noise rates span $p\in[0.001,0.020]$ over $11$ values,
and the shot budgets span the powers of two $B\in\{2^6,2^7,\ldots,2^{15}\}$;
the noise grid is fixed in the reproduction code.
Each cell estimates the empirical MSE of the noisy estimator, exact PEC under the
simulated noise model, and linear CDR.

\begin{figure}
\centering
\includegraphics[width=0.88\textwidth]{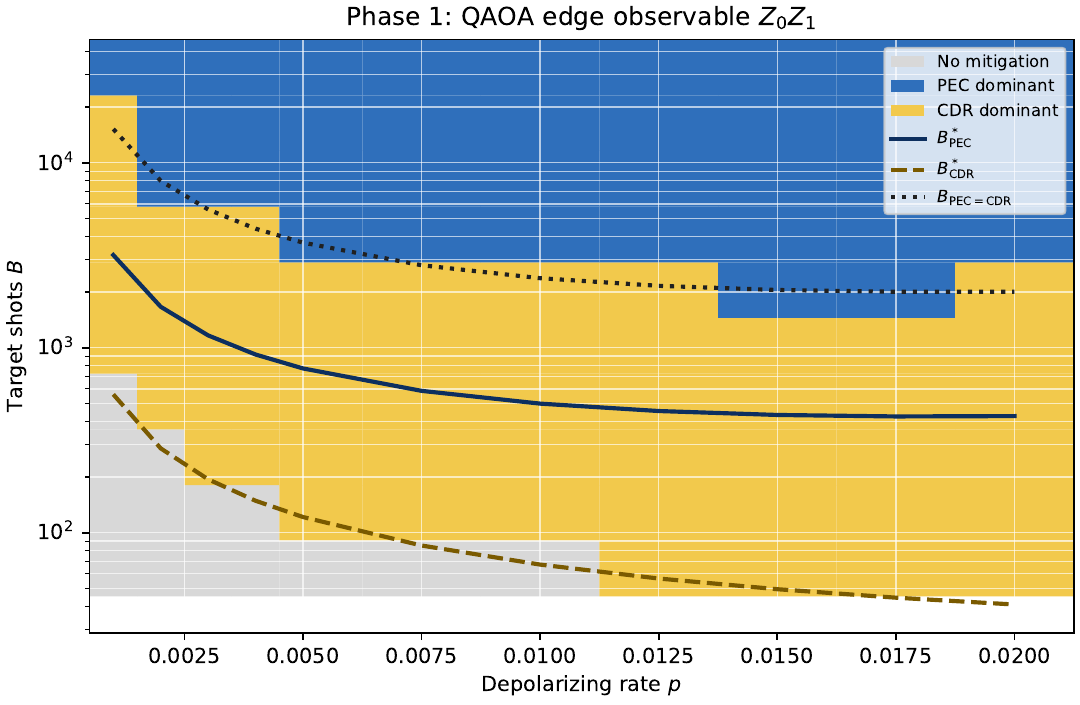}
\caption{Finite-shot QAOA comparison for the edge observable $Z_0Z_1$.  The
background color is the empirically MSE-optimal method at each simulated
$(p,B)$ cell.  The solid, dashed, and dotted curves show the predicted
$B_{\PEC}^*(p)$, $B_{\CDR}^*(p)$, and $B_{\PEC=\CDR}(p)$ boundaries computed from
the fitted population CDR model and exact PEC overhead.}
\label{fig:qaoa_validation_phase_map}
\end{figure}

The empirical comparison in Figure~\ref{fig:qaoa_validation_phase_map}
reproduces the finite-shot dominance structure predicted by
Theorem~\ref{thm:cdr_operating_window}.
At small budgets the noisy estimator can remain optimal; at intermediate budgets
CDR dominates because its target-shot variance overhead is smaller than PEC's;
and beyond the PEC-CDR crossover, exact PEC dominates because CDR retains a
first-order calibration residual.  Across the tested noise rates, the fitted
first-order residual is approximately stable,
$b_{\mathrm{mis}}(p)/p\approx 1.9$--$2.0$, and the crossover budgets fall from
about $1.52\times10^4$ shots at $p=0.001$ to about $2.01\times10^3$ shots at
$p=0.020$.  The drop is slower than the asymptotic $1/p$ law (a factor of
$7.6$ rather than $20$ across this range) because $\Gamma^2(p)$, and hence
$A_{\PEC}(p)-A_{\CDR}(p)$, grows superlinearly toward the upper end of the
tested noise range; the comparison is against the exact threshold
formulas \eqref{eq:pec_cdr_crossover_budget}, which track the empirical
crossovers within the validation tolerances.

\section{Numerical validation}
\label{sec:canonical_results}

The numerical study tests the finite-shot operating boundaries in progressively
less idealized settings.  The primary comparisons test the PEC/CDR/no-mitigation
regions, while the appendix analyses vary the observable, circuit instance,
noise channel, calibration design, and device-derived noise assumptions.  Table
\ref{tab:canonical_phase_results} summarizes the main numerical evidence for
the finite-shot result.

\begin{table}[b]
\centering
\footnotesize
\caption{Numerical summary for the finite-shot PEC/CDR operating
boundaries.}
\label{tab:canonical_phase_results}
\begin{ruledtabular}
\begin{tabular}{p{0.24\textwidth} p{0.68\textwidth}}
\toprule
Comparison & Main result \\
\colrule
Two-qubit example &
The analytic two-qubit calculation agrees with the closed-form threshold to maximum
relative error $2.46\times10^{-14}$, and all non-boundary region assignments
match. \\
QAOA edge observable &
The baseline edge observable shows the predicted finite-shot regions; all
non-boundary region assignments match, with boundary-local disagreement $0.0484$
and grid-cell ambiguity fraction $0.1091$. \\
Mismatch scaling &
The CDR crossover follows the predicted monotone dependence on calibration
mismatch; the median observed slope is $0.8959$ with IQR
$[0.8864,0.9160]$ over $34$ calibration rows. \\
Boundary bootstrap &
Bootstrap uncertainty is concentrated near region boundaries: the grid-cell
ambiguity fraction is $0.1091$, the boundary-ambiguity fraction is $0.1935$, and the mean
relative crossover error is $0.1481$. \\
MaxCut Hamiltonian &
The full commuting MaxCut Hamiltonian preserves the three-region structure; all
non-boundary region assignments match, with boundary-local disagreement $0.0370$. \\
Instance robustness &
Across $21$ graph instances and $1008$ Monte Carlo cells, all non-boundary
region assignments match; finite crossover errors are below $0.30$ in $0.9302$ of
cases and at most $0.50$ in $0.9767$ of cases. \\
Pauli channels &
The finite-shot prediction persists for the N2a, N2b, and N3 Pauli-channel
variants with channel-specific PEC inverse one-norms; all non-boundary region
assignments match, the finite-error fraction below threshold is $0.875$,
and boundary-local disagreement is $0.0230$. \\
\botrule
\end{tabular}
\end{ruledtabular}
\end{table}

Table~\ref{tab:canonical_phase_results} gives the compact numerical evidence for
the main finite-shot result.  Additional analyses are grouped by physical role:
robustness perturbations of the baseline QAOA setting, second-order,
response-projection, and shot-allocation tests of the analytic corollaries,
offline Aer/FakeBackend device-model simulations, and an archived real-device
count analysis.  The detailed plots and tables for these extensions are reported in
Appendices~\ref{app:canonical_numerical_details},
\ref{sec:phase7_supplement}, \ref{app:phase8_plus_experimental_layers}, and
\ref{app:device_model_hardware_counts}.

\paragraph{Synthesis.}
Taken together, the numerical comparisons show that the predicted finite-shot
boundaries are visible in the baseline QAOA setting and persist under several
controlled perturbations.  The archived real-device count study is reported
separately as fixed-count evidence rather than as a new hardware experiment.
\section{Conclusion}
\label{sec:conclusion}

We have derived finite-shot MSE boundaries for exact PEC, linear CDR, and no
mitigation under explicit assumptions.  The central mechanism is a finite
operating window: an unbiased estimator with larger variance overhead can be
outperformed, below a finite crossover, by a biased estimator with smaller
target-shot variance.  For PEC versus CDR, the upper endpoint scales as
$1/(\delta_1^2p)$, so the CDR-dominant region is highly sensitive to the
first-order CDR calibration mismatch.

The calibration results identify when this mismatch is structural rather than a
finite-sample artifact.  Population linear CDR is first-order calibrated exactly
on an affine response hyperplane, and the target-response projection theorem shows
that response-blind affine training has an irreducible ensemble-level mismatch
unless the target noise response is affine in the ideal target value.  This also
yields a local statistical-decision lower bound for response-blind affine CDR.  When
first-order calibration holds, the second-order residual produces the next
crossover scale, proportional to $1/(\delta_2^2p^3)$.  The same MSE bookkeeping
extends to commuting Pauli Hamiltonians, finite CDR training shots with an
explicit target/training allocation rule, and residual PEC model bias.

The numerical comparisons support the predicted no-mitigation, CDR-dominant,
and PEC-dominant regions in the baseline QAOA setting.  Additional robustness,
device-model, and archived-count analyses show how the same finite-shot
structure changes under variations of observable, circuit instance, noise
channel, training design, and device-derived noise assumptions.

The resulting picture is local but predictive.  For CDR, the relevant object is
the specified affine training family or target-conditioned design; improved
training sets can change the calibration mismatch and therefore shift the
PEC-CDR crossover.  The device-model simulations show that realistic PEC
noise-model violation can remove a PEC-dominant region in the tested
configuration, while the archived real-device count analysis provides
fixed-count evidence consistent with the device-facing interpretation.

\section*{Data and code availability}

The reproduction code, numerical outputs, simulator outputs,
generated figures, Docker environment, verification scripts, and reference
checksums are archived under the versioned release tag \path{v1.0.2} and at
\doi{10.5281/zenodo.20765589}.  The local verification route and Docker
route reproduce the baseline outputs and recheck the secondary robustness and
Aer/FakeBackend device-model artifacts.  The GitHub Actions workflow is intended
to exercise the same release checks on Linux, macOS 14, and Windows.

The real-device count analysis is treated as archived-count evidence.  The
repository contains the frozen analysis JSON files and SHA-256 anchors.  The
four raw count files are distributed in the same Zenodo record.  Public
reproduction of the archived-count analysis consists of rerunning the analysis
from the archived counts and verifying their checksums.

The public repository is available at
\url{https://github.com/vicenzoscavino1999/finite-shot-pec-cdr}.  Code is
released under the MIT license; archived numerical data and generated figures
are released under CC-BY-4.0.

\section*{Additional statements}

\textbf{Funding.} The author received no funding for this research.

\textbf{Competing interests.} The author declares no competing interests.

\textbf{Author contributions.} Vicenzo Scavino is the sole author and is
responsible for the conception, argument, implementation, numerical validation,
drafting, revision, and final approval of the manuscript.

\textbf{AI assistance disclosure.} AI-assisted tools were used for language
editing, formatting support, code organization, and non-substantive structural
refinement.  The author reviewed, revised, verified, and approved the final
manuscript and accepts full responsibility for its content.

\begin{acknowledgments}
The author is deeply grateful to his parents, Luigi Stefano Scavino Montero
and Isabel Betzabe Alfaro Ninahuanca, and to his sister, Luciana Franceska
Scavino Alfaro, for their unwavering personal support and encouragement
throughout this work.
\end{acknowledgments}

\clearpage

\appendix

\section{Full help-harm sign cases}
\label{app:help_harm_cases}

The gain formula \eqref{eq:Delta_general} gives the following complete sign
classification:
\begin{center}
\small
\refstepcounter{table}
\label{tab:help_harm_sign_cases}
TABLE~\thetable. Complete help-harm sign cases for the generic finite-shot gain
\eqref{eq:Delta_general}.

\vspace{0.5em}
\begin{ruledtabular}
\begin{tabular}{c c c}
\toprule
Condition on $D_M$ & Condition on $A_M$ & Operating implication \\
\colrule
$D_M>0$ & $A_M>0$ & helps iff $B>A_M/D_M$ \\
$D_M>0$ & $A_M\le0$ & helps for every $B>0$ \\
$D_M=0$ & $A_M>0$ & harms for every $B>0$ \\
$D_M=0$ & $A_M=0$ & exact tie for every $B>0$ \\
$D_M=0$ & $A_M<0$ & helps for every $B>0$ \\
$D_M<0$ & $A_M\ge0$ & harms for every $B>0$ \\
$D_M<0$ & $A_M<0$ & helps iff $0<B<A_M/D_M$ \\
\botrule
\end{tabular}
\end{ruledtabular}
\end{center}
The last ratio is positive because both numerator and denominator are negative.

\clearpage

\section{Detailed CDR calibration and finite-training results}
\label{app:cdr_mathematical_details}

This appendix records the CDR coefficient expansion, response-projection results,
and finite-training bounds summarized in Section~\ref{sec:cdr_calibration}.  The
estimator class and assumptions are the same as in the main text; only the
algebraic and statistical details are deferred here.

\begin{lemma}[CDR response expansion]
\label{lem:cdr_expansion}
Under Assumption~\ref{ass:cdr_response}, let
\begin{align}
    V_y&=\Var_{\mathcal{T}}(y)>0,
    &
    C_{y\beta}&=\Cov_{\mathcal{T}}(y,\beta),\\
    C_{y\chi}&=\Cov_{\mathcal{T}}(y,\chi),
    &
    V_\beta&=\Var_{\mathcal{T}}(\beta),
\end{align}
\begin{align}
    \bar y&=\E_{\mathcal{T}}[y],
    &
    \bar\beta&=\E_{\mathcal{T}}[\beta],
    &
    \bar\chi&=\E_{\mathcal{T}}[\chi].
\end{align}
Then
\begin{align}
    a^*(p)&=1-a_1p+a_2p^2+O(p^3),\label{eq:a_expansion}\\
    b^*(p)&=-b_1p-b_2p^2+O(p^3),\label{eq:b_expansion}
\end{align}
where
\begin{equation}
    a_1=\frac{C_{y\beta}}{V_y},
    \qquad
    b_1=\bar\beta-a_1\bar y,
    \label{eq:a1b1}
\end{equation}
\begin{equation}
    a_2=2a_1^2-\frac{C_{y\chi}+V_\beta}{V_y},
    \qquad
    b_2=\bar\chi-a_1\bar\beta+a_2\bar y.
    \label{eq:a2b2}
\end{equation}
Consequently,
\begin{equation}
    b_{\mathrm{mis}}(p)=\delta_1p+\delta_2p^2+O(p^3),
    \qquad
    \delta_1=\beta_T-a_1\mu_0-b_1.
    \label{eq:delta1}
\end{equation}
with
\begin{equation}
    \delta_2=\chi_T-a_1\beta_T+a_2\mu_0-b_2.
    \label{eq:delta2}
\end{equation}
\end{lemma}

\begin{proof}
Population least squares gives
$a^*(p)=\Cov(x(p),y)/\Var(x(p))$ and
$b^*(p)=\E[y]-a^*(p)\E[x(p)]$.  Since
$x(p)=y+\beta p+\chi p^2+O(p^3)$,
\begin{align}
    \Cov(x(p),y)&=V_y+C_{y\beta}p+C_{y\chi}p^2+O(p^3),\\
    \Var(x(p))&=V_y+2C_{y\beta}p+(2C_{y\chi}+V_\beta)p^2+O(p^3).
\end{align}
Expanding the ratio gives \eqref{eq:a_expansion}.  The intercept expansion
follows from $b^*(p)=\bar y-a^*(p)(\bar y+\bar\beta p+\bar\chi p^2+O(p^3))$.
Substitution into \eqref{eq:bmis_def} gives \eqref{eq:delta1} and
\eqref{eq:delta2}.
\end{proof}

\begin{definition}[First-order calibrated CDR]
\label{def:first_order_calibrated}
The training distribution is first-order calibrated for the target if
\begin{equation}
    \delta_1=\beta_T-a_1\mu_0-b_1=0.
    \label{eq:first_order_calibration}
\end{equation}
Equivalently, if $\beta_{\mathrm{eff,tr}}(\mu_0)=a_1\mu_0+b_1$, calibration means
$\beta_T=\beta_{\mathrm{eff,tr}}(\mu_0)$.
\end{definition}

\begin{theorem}[First-order CDR calibration and OLS necessity]
\label{thm:cdr_generic}
Fix a training distribution satisfying Assumption~\ref{ass:cdr_response}.
\begin{enumerate}
    \item In the two-dimensional target-response coordinates
    $(\mu_0,\beta_T)$, the targets for which linear CDR is first-order
    calibrated form the affine hyperplane
\begin{equation}
    \mathcal{C}_{\mathrm{cal}}
    =\{(\mu_0,\beta_T):\beta_T=a_1\mu_0+b_1\}.
    \label{eq:calibration_hyperplane}
\end{equation}
    For every target outside this set,
\begin{equation}
    b_{\mathrm{mis}}(p)=\delta_1p+O(p^2),
    \qquad
    \delta_1\ne0.
\end{equation}
    Thus an $O(p^2)$ CDR residual is a structural calibration condition on the
    relation between the training distribution and the target.
    \item For the population OLS predictor
    $f_p(x)=a^*(p)x+b^*(p)$ learned from this fixed training distribution,
    the target residual satisfies
\begin{equation}
    f_p(\mu_p)-\mu_0=O(p^2)
\end{equation}
    if and only if the target is first-order calibrated,
    $\delta_1=\beta_T-a_1\mu_0-b_1=0$.
\end{enumerate}
The second statement applies to OLS-CDR trained from the fixed distribution
$\mathcal T$; target-specific affine maps belong to a different design class.
\end{theorem}

\begin{proof}
Lemma~\ref{lem:cdr_expansion} gives
$\delta_1=\beta_T-a_1\mu_0-b_1$.  Therefore $\delta_1=0$ is exactly the affine
constraint \eqref{eq:calibration_hyperplane}.  Its complement is open and dense
in the two-dimensional target-response coordinates, and the hyperplane has
Lebesgue measure zero.  The same lemma gives
$f_p(\mu_p)-\mu_0=\delta_1p+O(p^2)$ for the population OLS predictor.
If $\delta_1=0$, the residual is $O(p^2)$.  Conversely, if the residual is
$O(p^2)$, the coefficient of $p$ in this expansion must vanish, so
$\delta_1=0$.
\end{proof}

\begin{corollary}[Second-order CDR calibration hierarchy]
\label{cor:second_order_cdr}
Assume the first-order calibration condition $\delta_1=0$ holds for the target.
If $\delta_2\ne0$, then population OLS-CDR has residual bias
\begin{equation}
    b_{\mathrm{mis}}(p)=\delta_2p^2+O(p^3).
\end{equation}
If, in addition,
\begin{equation}
    A_{\PEC}(p)-A_{\CDR}(p)
    =(\kappa_{\PEC}-\kappa_{\CDR})p+O(p^2),
    \qquad \kappa_{\PEC}>\kappa_{\CDR},
\end{equation}
then the PEC-CDR crossover obeys
\begin{equation}
    B_{\PEC=\CDR}^{(2)}(p)
    =
    \frac{\kappa_{\PEC}-\kappa_{\CDR}}{\delta_2^2}\frac{1}{p^3}
    +O(p^{-2}).
    \label{eq:second_order_crossover}
\end{equation}
Under first-order calibration, the finite-shot dominance pattern is controlled by
the next nonzero response coefficient.
\end{corollary}

\begin{proof}
The residual expansion follows from Lemma~\ref{lem:cdr_expansion}.  The
PEC-CDR crossover is
$B_{\PEC=\CDR}=(A_{\PEC}-A_{\CDR})/b_{\mathrm{mis}}^2$, derived in
\eqref{eq:pec_cdr_crossover_budget} below.  Substituting
$b_{\mathrm{mis}}^2=\delta_2^2p^4+O(p^5)$ and the stated variance-overhead
expansion gives \eqref{eq:second_order_crossover}.
\end{proof}

The second-order crossover scale \eqref{eq:second_order_crossover} is a
falsifiable prediction of the calibration hierarchy.  The primary numerical
comparisons use nonzero first-order mismatch, while the supplementary
second-order tests probe this regime directly.

\begin{proposition}[Calibration/noise-response circularity]
\label{prop:cdr_circularity}
For a fixed population linear CDR training distribution, first-order calibration
of a target is equivalent to matching the target's first-order noise response to
the affine response learned from training:
\begin{equation}
    \beta_T=\beta_{\mathrm{eff,tr}}(\mu_0)=a_1\mu_0+b_1.
\end{equation}
Consequently, a fixed training distribution satisfies $\delta_1=0$
simultaneously for every target in a class $\mathcal{C}$ of response pairs
$(\mu_0,\beta_T)$ if and only if $\mathcal{C}$ is contained in the single
affine line $\{\beta_T=a_1\mu_0+b_1\}$.  Equivalently, simultaneous first-order
calibration by one fixed training distribution requires collinearity of the
target response pairs.
\end{proposition}

\begin{proof}
The equivalence is exactly Definition~\ref{def:first_order_calibrated} and
Theorem~\ref{thm:cdr_generic}.  Simultaneous calibration for every target in
$\mathcal{C}$ means that every pair $(\mu_0,\beta_T)\in\mathcal{C}$ lies on the
affine calibration set \eqref{eq:calibration_hyperplane}, which is a single
affine line determined by $(a_1,b_1)$.
\end{proof}

\begin{remark}
This is the conceptual tension behind CDR calibration.  CDR is attractive partly
because it avoids an explicit noise inverse, but first-order calibrated CDR
requires the training family to carry the same first-order noise-response
information that determines the target's leading noisy drift: a
target-conditioned design that guarantees $\delta_1=0$ for a noncollinear
target class must encode that information, either analytically through a noise
model or empirically through additional calibration data.  Successful training
sets of this kind therefore express a calibration property rather than a generic
consequence of using a linear regression.
\end{remark}

\begin{theorem}[Target-response projection theorem for CDR mismatch]
\label{thm:response_projection}
Let targets be drawn from an ensemble $\Pi$ over first-order response pairs
$(\mu_0,\beta_T)$ with finite second moments.  Any fixed population linear CDR
training design induces an affine first-order response
\begin{equation}
    h(\mu)=a_1\mu+b_1
\end{equation}
and hence ensemble-average squared first-order mismatch
\begin{equation}
    \mathcal E(h)=
    \E_{\Pi}\!\left[(\beta_T-h(\mu_0))^2\right].
\end{equation}
Among all affine first-order responses, the optimal mismatch is
\begin{equation}
    \mathcal E_{\mathrm{aff}}
    =
    \inf_{a,b}\E_{\Pi}\!\left[(\beta_T-a\mu_0-b)^2\right].
    \label{eq:affine_response_projection_error}
\end{equation}
Every response-blind training design therefore satisfies
$\mathcal E(h)\ge\mathcal E_{\mathrm{aff}}$; attaining
$\mathcal E_{\mathrm{aff}}$ requires the $L^2(\Pi)$ projection minimizer to be
realizable by an admissible training family.
Equivalently,
\begin{equation}
    \mathcal E_{\mathrm{aff}}
    =
    \E_{\Pi}\!\left[\Var(\beta_T\mid \mu_0)\right]
    +
    \inf_{a,b}\E_{\Pi}\!\left[
    \left(\E[\beta_T\mid\mu_0]-a\mu_0-b\right)^2
    \right].
    \label{eq:response_projection_decomposition}
\end{equation}
Thus $\mathcal E_{\mathrm{aff}}=0$ if and only if
$\beta_T=a\mu_0+b$ almost surely for some affine function.  Outside this
condition, every response-blind population linear CDR design has a nonzero
ensemble-average first-order bias floor.  Moreover, if the conditional law of
$\beta_T$ given $\mu_0$ is non-atomic for $\Pi$-almost every $\mu_0$, then for
any fixed affine response $h$,
\begin{equation}
    \Pr_{\Pi}\{\delta_1=0\}=
    \Pr_{\Pi}\{\beta_T=h(\mu_0)\}=0.
\end{equation}
\end{theorem}

\begin{proof}
For a fixed training design, Lemma~\ref{lem:cdr_expansion} gives
$\delta_1=\beta_T-h(\mu_0)$, so the ensemble-average squared first-order
mismatch is $\mathcal E(h)$.  Minimizing over affine first-order responses is
exactly the $L^2(\Pi)$ projection of $\beta_T$ onto the affine span of
$(1,\mu_0)$, which gives \eqref{eq:affine_response_projection_error}.  Applying
the law of total variance to
$\beta_T-a\mu_0-b$ gives
\eqref{eq:response_projection_decomposition}.  The right-hand side is zero only
when both terms vanish: first, $\beta_T=\E[\beta_T\mid\mu_0]$ almost surely, and
second, this conditional mean is affine in $\mu_0$.  The final claim follows
because, conditional on such a $\mu_0$, the event
$\{\beta_T=h(\mu_0)\}$ is a single point under a non-atomic conditional law.
\end{proof}

\begin{corollary}[Local Bayes/minimax lower bound for response-blind affine CDR]
\label{cor:bayes_minimax_cdr}
Let $\mathcal{T}$ be a compact target class and let $\Pi$ be any distribution
supported on $\mathcal{T}$.  Consider the response-blind affine class
\begin{equation}
    \mathcal{H}_{\mathrm{aff}}(M)
    =\{h(\mu)=a\mu+b:\ |a|+|b|\le M\}
\end{equation}
of first-order responses, with $M$ chosen large enough to contain an
$L^2(\Pi)$ projection minimizer defining
$\mathcal E_{\mathrm{aff}}$ in
Theorem~\ref{thm:response_projection}.  The constraint $|a|+|b|\le M$ makes the
expansion remainders below uniform over the class; containing a projection
minimizer makes the bound asymptotically tight, while the displayed
inequalities hold for every $M>0$ because
$\mathcal E(h)\ge\mathcal E_{\mathrm{aff}}$ for every affine $h$.  Assume the
expansions above hold uniformly on $\mathcal{T}$.  Let
\begin{equation}
    \bar\sigma_0^2=\E_{\Pi}[1-\mu_0^2].
\end{equation}
For $h(\mu)=a\mu+b\in\mathcal{H}_{\mathrm{aff}}(M)$, let the truncated local CDR
estimator be
\begin{equation}
    f_p(x)=(1-ap)x-bp,
\end{equation}
applied to the target sample mean, and write $\MSE_h(p,B)$ for its target-shot
MSE.  Then there is a constant $C$ independent of $p$, $B$, and
$h\in\mathcal{H}_{\mathrm{aff}}(M)$ such that, for sufficiently small $p$,
\begin{equation}
    \inf_{h\in\mathcal{H}_{\mathrm{aff}}(M)}
    \E_{\Pi}[\MSE_h(p,B)]
    \ge
    \mathcal E_{\mathrm{aff}}p^2+\frac{\bar\sigma_0^2}{B}
    -C\!\left(p^3+\frac{p}{B}\right).
    \label{eq:bayes_cdr_lower_bound}
\end{equation}
Consequently,
\begin{equation}
    \inf_{h\in\mathcal{H}_{\mathrm{aff}}(M)}
    \sup_{T\in\mathcal{T}}\MSE_h(p,B)
    \ge
    \mathcal E_{\mathrm{aff}}p^2+\frac{\bar\sigma_0^2}{B}
    -C\!\left(p^3+\frac{p}{B}\right).
    \label{eq:minimax_cdr_lower_bound}
\end{equation}
\end{corollary}

\begin{proof}
For a fixed $h\in\mathcal{H}_{\mathrm{aff}}(M)$, the bias of the truncated
estimator is
\begin{equation}
    b_{\mathrm{mis},h}(p)
    =(1-ap)\mu_p-bp-\mu_0
    =(\beta_T-h(\mu_0))p+O(p^2),
\end{equation}
with the $O(p^2)$ constant uniform on $\mathcal{T}$ and over
$\mathcal{H}_{\mathrm{aff}}(M)$, because $|a|+|b|\le M$ and the target
expansions are uniform on the compact class.  The slope satisfies
$1-ap=1+O(p)$ uniformly, so the target-shot variance term is
$(1-\mu_0^2)/B+O(p/B)$, again uniformly.  Averaging over $\Pi$ gives
\begin{equation}
    \E_{\Pi}[\MSE_h(p,B)]
    =
    \E_{\Pi}[(\beta_T-h(\mu_0))^2]p^2
    +\frac{\bar\sigma_0^2}{B}
    +O\!\left(p^3+\frac{p}{B}\right).
\end{equation}
Taking the infimum over $h\in\mathcal H_{\mathrm{aff}}(M)$ and using
$\mathcal E(h)\ge\mathcal E_{\mathrm{aff}}$ gives the Bayes-risk bound.  Since
$\Pi$ is supported on $\mathcal{T}$, $\sup_{T\in\mathcal{T}}\MSE_h\ge
\E_{\Pi}[\MSE_h]$ for each $h$, and taking the infimum gives
\eqref{eq:minimax_cdr_lower_bound}.
\end{proof}

\begin{remark}
Corollary~\ref{cor:bayes_minimax_cdr} applies to response-blind affine
mitigation.  Nonlinear, adaptive, and target-informed designs form different
hypothesis classes.  A target-informed CDR design that incorporates the target's
first-order response can achieve $\delta_1=0$ and lies outside this
response-blind class; the lower bound quantifies the value of target-response
information.
\end{remark}

\begin{remark}[Training-design classes]
\label{rem:training_design_classes}
Theorem~\ref{thm:response_projection} gives a quantitative version of
``generic mismatch.''  A response-blind training design fixes an affine response
$h$ and carries the projection floor $\mathcal E_{\mathrm{aff}}$.  A
structure-matched design can reduce the first term empirically by making the
training circuits' response closer to the target family; the affine projection
term is removed when the target response is affine in $\mu_0$ on the tested
family.  A target-response-informed design can set $\delta_1=0$ for a chosen
target by supplying the response information $\beta_T-h(\mu_0)$, either from a
noise model, additional calibration experiments, or an explicitly engineered
training family.  This is the sense in which first-order calibrated CDR is
possible with target-response information in the first-order response variable.
\end{remark}

\begin{proposition}[Affine-reparametrization converse]
\label{thm:affine_converse}
Let $x\mapsto x'=\alpha x+\eta$, $\alpha\ne0$, be any invertible affine
reparametrization of the noisy expectation coordinate, applied consistently to
training and target circuits.  The population linear CDR prediction is invariant:
\begin{equation}
    a'^{*}(p)(\alpha\mu_{T,p}+\eta)+b'^{*}(p)-\mu_0
    =a^*(p)\mu_{T,p}+b^*(p)-\mu_0 .
\end{equation}
Consequently, if $\delta_1\ne0$, no invertible affine rescaling or shifting of
$\mu_p$ can make linear CDR first-order calibrated for that target.
\end{proposition}

\begin{proof}
OLS predictions are invariant under an invertible affine change of the predictor
coordinate: the transformed predictor is
$f'_p(x')=(a^*(p)/\alpha)x'+b^*(p)-a^*(p)\eta/\alpha$, so
$f'_p(\alpha x+\eta)=a^*(p)x+b^*(p)$ for every $x$.  Since the argument is
pointwise in $p$, the same invariance holds for $p$-dependent
reparametrizations $x\mapsto\alpha(p)x+\eta(p)$ with $\alpha(p)\ne0$.
\end{proof}

\begin{assumption}[Linearized finite-training CDR]
\label{ass:finite_training_cdr}
Let $\theta^*(p)=(a^*(p),b^*(p))^\top$ be the population CDR coefficients and
let $\hat\theta(p)=(\hat a(p),\hat b(p))^\top$ be the coefficients learned from
$N_t$ training circuits, each measured with $B_{\mathrm{train}}$ shots.  In the
local linearized regime of the OLS map, assume the coefficient error
$e=\hat\theta-\theta^*$ is independent of the target-shot noise, has zero mean,
and satisfies
\begin{equation}
    \E[\norm{e}_{G}^2]\le
    \frac{C_{\mathrm{tr}}}{N_tB_{\mathrm{train}}},
    \qquad
    G\succeq \lambda_0 I,\quad \lambda_0>0,
    \label{eq:finite_training_coeff_bound}
\end{equation}
where $\norm{e}_G^2=e^\top Ge$ and $G$ is the population Gram matrix of the
affine training features $(x_j(p),1)$.
\end{assumption}

\begin{proposition}[Conditioning amplification of finite-training noise]
\label{prop:conditioning_amplification}
Let $z_T=(\mu_p,1)^\top$ be the affine target feature vector and let
$e=\hat\theta-\theta^*$ be the finite-training coefficient error.  If, in the
same local regime as Assumption~\ref{ass:finite_training_cdr}, the coefficient
covariance has the lower proxy
\begin{equation}
    \E[ee^\top]\succeq
    \frac{c_{\mathrm{tr}}}{N_tB_{\mathrm{train}}}G^{-1},
    \qquad c_{\mathrm{tr}}>0,
    \label{eq:finite_training_cov_lower}
\end{equation}
then the coefficient-noise contribution to the target MSE obeys
\begin{equation}
    \E[(z_T^\top e)^2]
    \ge
    \frac{c_{\mathrm{tr}}}{N_tB_{\mathrm{train}}}
    z_T^\top G^{-1}z_T.
    \label{eq:conditioning_lower_bound}
\end{equation}
In particular, if the bound $G\succeq\lambda_0 I$ is attained, so that
$v_{\min}$ is a unit eigenvector of $G$ with smallest eigenvalue $\lambda_0$,
and $\rho=z_T^\top v_{\min}$, then
\begin{equation}
    \E[(z_T^\top e)^2]
    \ge
    \frac{c_{\mathrm{tr}}\rho^2}
    {\lambda_0N_tB_{\mathrm{train}}}.
    \label{eq:conditioning_lambda_bound}
\end{equation}
\end{proposition}

\begin{proof}
The coefficient-noise contribution is
\begin{equation}
    \E[(z_T^\top e)^2]=z_T^\top \E[ee^\top] z_T.
\end{equation}
Substituting \eqref{eq:finite_training_cov_lower} gives
\eqref{eq:conditioning_lower_bound}.  Expanding
$z_T^\top G^{-1}z_T$ in the eigenbasis of $G$ and keeping only the component
along $v_{\min}$ gives \eqref{eq:conditioning_lambda_bound}.
\end{proof}

\begin{remark}
Proposition~\ref{prop:conditioning_amplification} shows that the
$1/\lambda_0$ dependence in the finite-training analysis reflects an operational
effect rather than a proof artifact.  If the training feature Gram matrix is
nearly singular in a direction
seen by the target feature vector, finite-training noise is amplified at the
target even when the population CDR map is well defined.
\end{remark}

\begin{proposition}[Finite-training CDR correction]
\label{prop:finite_training_cdr_bound}
Under Assumption~\ref{ass:finite_training_cdr}, let the target noisy sample mean
be $\hat\mu_p=\mu_p+\xi_T$, with
$\E[\xi_T]=0$ and $\Var(\xi_T)=\sigma_N^2(p)/B$.  The finite-training CDR
estimator
\begin{equation}
    \hat X_{\CDR}=\hat a(p)\hat\mu_p+\hat b(p)
\end{equation}
satisfies the fully explicit upper bound
\begin{equation}
    \E[(\hat X_{\CDR}-\mu_0)^2]
    \le
    b_{\mathrm{mis}}^2(p)
    +a^{*2}(p)\frac{\sigma_N^2(p)}{B}
    +\frac{2C_{\mathrm{tr}}}{\lambda_0N_tB_{\mathrm{train}}}
    +\frac{C_{\mathrm{tr}}\,\sigma_N^2(p)}
    {\lambda_0N_tB_{\mathrm{train}}B}.
    \label{eq:finite_training_cdr_bound_formula}
\end{equation}
\end{proposition}

\begin{proof}
Write $z_T=(\mu_p,1)^\top$.  Since $|\mu_p|\le1$ for a Pauli expectation,
$\norm{z_T}^2\le2$.  Expanding the estimator gives
\begin{equation}
    \hat X_{\CDR}-\mu_0
    =
    b_{\mathrm{mis}}(p)+z_T^\top e+a^*(p)\xi_T+e_a\xi_T,
\end{equation}
where $e_a$ is the slope component of $e$.  The zero-mean and independence
assumptions remove the first-order cross terms.  Moreover,
\begin{equation}
    \E[(z_T^\top e)^2]
    \le
    \norm{z_T}^2\,\E[\norm{e}^2]
    \le
    \frac{2}{\lambda_0}\E[\norm{e}_G^2]
    \le
    \frac{2C_{\mathrm{tr}}}{\lambda_0N_tB_{\mathrm{train}}}.
\end{equation}
The target-shot term contributes
$a^{*2}\sigma_N^2/B$.  For the product term, independence gives
$\E[e_a^2\xi_T^2]=\E[e_a^2]\,\sigma_N^2/B$, and
$\E[e_a^2]\le\E[\norm{e}^2]\le
C_{\mathrm{tr}}/(\lambda_0N_tB_{\mathrm{train}})$ under
\eqref{eq:finite_training_coeff_bound}, which gives the last term.
\end{proof}

\begin{corollary}[Optimal target/training shot allocation for the CDR bound]
\label{cor:optimal_shot_allocation}
Ignore the higher-order product term in
\eqref{eq:finite_training_cdr_bound_formula} and define
\begin{equation}
    V_T(p)=a^{*2}(p)\sigma_N^2(p),
    \qquad
    K_{\mathrm{tr}}=\frac{2C_{\mathrm{tr}}}{\lambda_0}.
\end{equation}
For a fixed total CDR shot budget
\begin{equation}
    B_{\mathrm{tot}}=B+N_tB_{\mathrm{train}},
\end{equation}
the finite-shot part of the upper bound,
\begin{equation}
    \frac{V_T(p)}{B}
    +\frac{K_{\mathrm{tr}}}{N_tB_{\mathrm{train}}},
\end{equation}
is minimized by
\begin{align}
    B^*&=
    \frac{\sqrt{V_T(p)}}{\sqrt{V_T(p)}+\sqrt{K_{\mathrm{tr}}}}\,
    B_{\mathrm{tot}},\\
    B_{\mathrm{train}}^*&=
    \frac{\sqrt{K_{\mathrm{tr}}}}
    {N_t(\sqrt{V_T(p)}+\sqrt{K_{\mathrm{tr}}})}\,
    B_{\mathrm{tot}}.
    \label{eq:optimal_training_shots}
\end{align}
The minimized finite-shot correction is
\begin{equation}
    \frac{(\sqrt{V_T(p)}+\sqrt{K_{\mathrm{tr}}})^2}{B_{\mathrm{tot}}}.
    \label{eq:optimal_allocation_value}
\end{equation}
\end{corollary}

\begin{proof}
Set $Y=N_tB_{\mathrm{train}}$.  The constraint is $B+Y=B_{\mathrm{tot}}$, and
the bound to minimize is $V_T/B+K_{\mathrm{tr}}/Y$, which is strictly convex on
$0<B<B_{\mathrm{tot}}$ in the nondegenerate case.  The first-order condition is
$V_T/B^2=K_{\mathrm{tr}}/Y^2$, so
$B:Y=\sqrt{V_T}:\sqrt{K_{\mathrm{tr}}}$.  Substitution gives
\eqref{eq:optimal_training_shots} and \eqref{eq:optimal_allocation_value}.
\end{proof}

\begin{remark}
Corollary~\ref{cor:optimal_shot_allocation} optimizes the finite-training upper
bound.  Operationally, it makes the conditioning tradeoff explicit: larger
$K_{\mathrm{tr}}$ shifts shots from the target circuit into the training
circuits, while larger target-shot variance $V_T$ shifts shots back to the
target circuit.
\end{remark}

\begin{remark}
Proposition~\ref{prop:finite_training_cdr_bound} gives a local finite-sample
upper bound under the standard linearized, well-conditioned, sub-Gaussian
coefficient-error regime.  Its purpose is to show the scaling of the additional
finite-training correction: the population CDR dominance interval survives when
$N_tB_{\mathrm{train}}$ is large enough that this term is below the
calibration-bias and target-shot terms.
\end{remark}

\begin{remark}[Limits of the finite-training assumption]
\label{rem:finite_training_limits}
Assumption~\ref{ass:finite_training_cdr} describes the well-conditioned local
finite-training regime.  Departures from this regime arise with very small
$N_t$, very low $B_{\mathrm{train}}$, a nearly singular training Gram matrix,
target feature vectors outside the support of the training distribution, or
predictor noise large enough that the local OLS linearization becomes an
errors-in-variables problem rather than a small coefficient perturbation.  In
those settings the relevant tests are the conditioning of $G$ and the
empirical stability of $(\hat a,\hat b)$ under repeated finite-training
resampling.
\end{remark}
\section{Additional numerical validation}
\label{app:canonical_numerical_details}

\subsection{Mismatch scaling and boundary uncertainty}
\label{subsec:mismatch_boundary_results}

The mismatch-scaling experiment supports the leading-order prediction that
CDR's useful budget interval is controlled by calibration mismatch.  Each
calibration row fits $\log B_{\PEC=\CDR}$ against $\log(1/\delta_1^2)$ at fixed
$p$, so the leading-order law \eqref{eq:cdr_window_upper} predicts slope one.
The median observed slope is $0.8959$ with IQR $[0.8864,0.9160]$ across the
specified rows.  The finite-grid estimate remains stable across rows and tracks
the predicted exponent-one mismatch trend.  The bootstrap analysis shows that
uncertainty is localized near region boundaries: all qualitative regions remain
stable, the grid-cell ambiguity fraction is $0.1091$, and the boundary
ambiguous fraction is $0.1935$.

\begin{figure}
\centering
\includegraphics[width=0.78\textwidth]{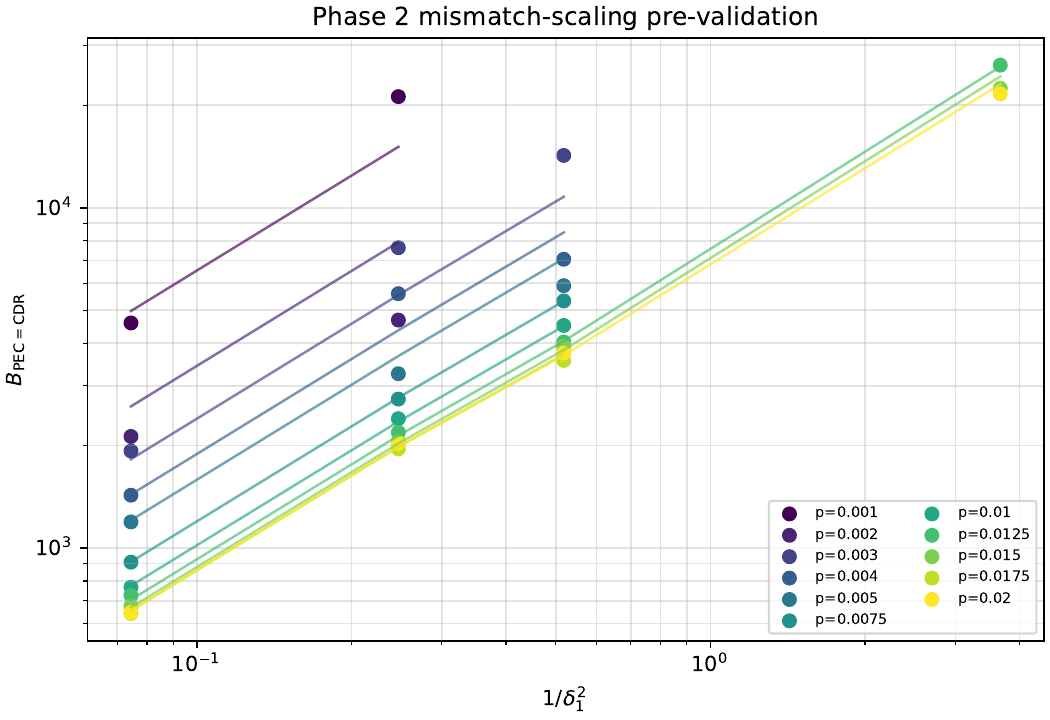}
\caption{Mismatch-scaling validation.  The observed crossover scale follows the
leading-order trend with median slope $0.8959$ over the calibration rows.}
\label{fig:phase2_mismatch_scaling}
\end{figure}

\begin{figure}
\centering
\includegraphics[width=0.70\textwidth]{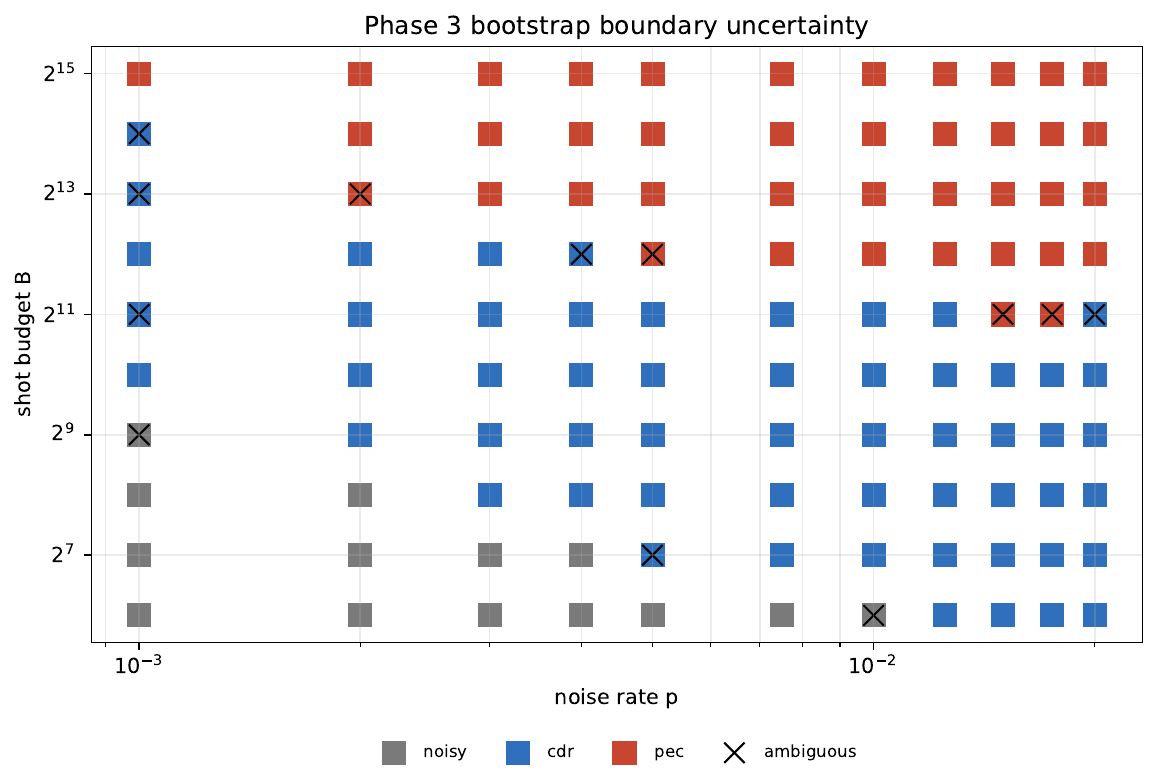}
\caption{Bootstrap boundary test.  Ambiguity is concentrated near the
predicted phase boundaries rather than spread across the non-boundary regions.}
\label{fig:phase3_bootstrap_boundary}
\end{figure}

\subsection{Hamiltonian and instance robustness}
\label{subsec:hamiltonian_instance_results}

The baseline scalar edge observable is complemented by a Hamiltonian-level
test, which replaces $Z_0Z_1$ with the full commuting weighted MaxCut
Hamiltonian and includes the covariance terms from
Section~\ref{sec:hamiltonian_extension}.  The Hamiltonian CDR map uses the
scalar-trained Hamiltonian correction of
Section~\ref{sec:hamiltonian_extension}.  All non-boundary classifications
match, with boundary-local disagreement $0.0370$.  The largest observed-fit
relative threshold error is
$0.4494$, again at a boundary or low-noise point.

\begin{figure}
\centering
\includegraphics[width=0.78\textwidth]{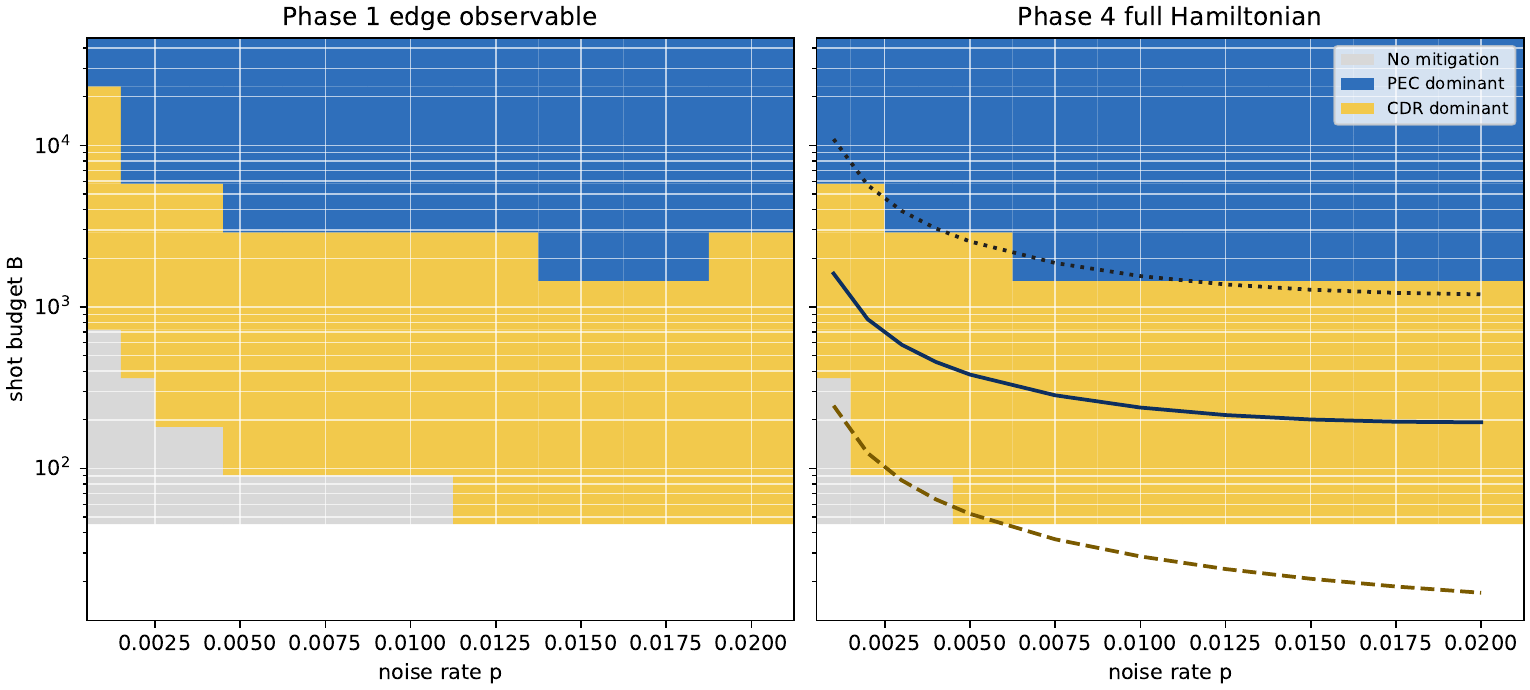}
\caption{Comparison between the baseline edge observable and the full
commuting MaxCut Hamiltonian validation.  The three-region structure survives
the covariance-level Hamiltonian upgrade.}
\label{fig:phase4_edge_vs_hamiltonian}
\end{figure}

The instance-robustness sweep then tests selected QAOA instances across graph
families and graph sizes up to $n=10$.  The full sweep covers $21$
instances and $1008$ Monte Carlo cells.  All non-boundary classifications match; the relative-error
fraction below $0.30$ is $0.9302$, and the fraction at most $0.50$ is $0.9767$.
The median relative crossover error is $0.0303$.  The
$n=10$ tier is included; it contributes $18$ rows, $12$ finite crossover
errors, and a $<0.30$ fraction of $0.9167$.  The finite-crossing set includes one
boundary-local point at relative error $0.4370$.

\begin{figure}
\centering
\includegraphics[width=0.78\textwidth]{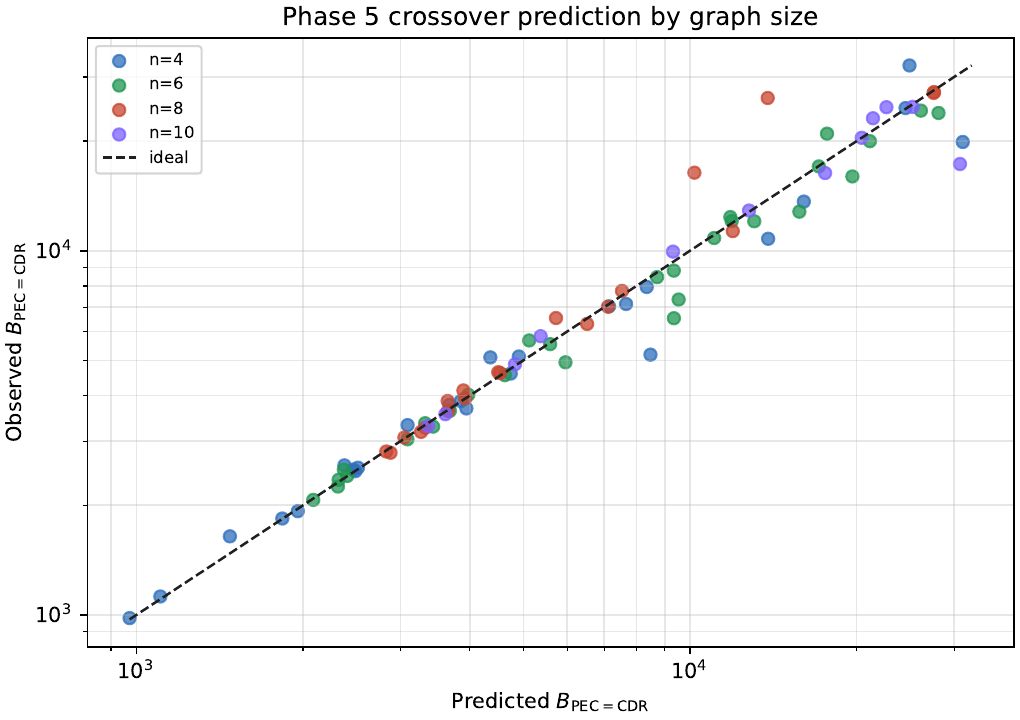}
\caption{Crossover prediction by graph size for the selected
instance-robustness sweep.}
\label{fig:phase5_crossover_by_graph_size}
\end{figure}

\begin{figure}
\centering
\includegraphics[width=0.78\textwidth]{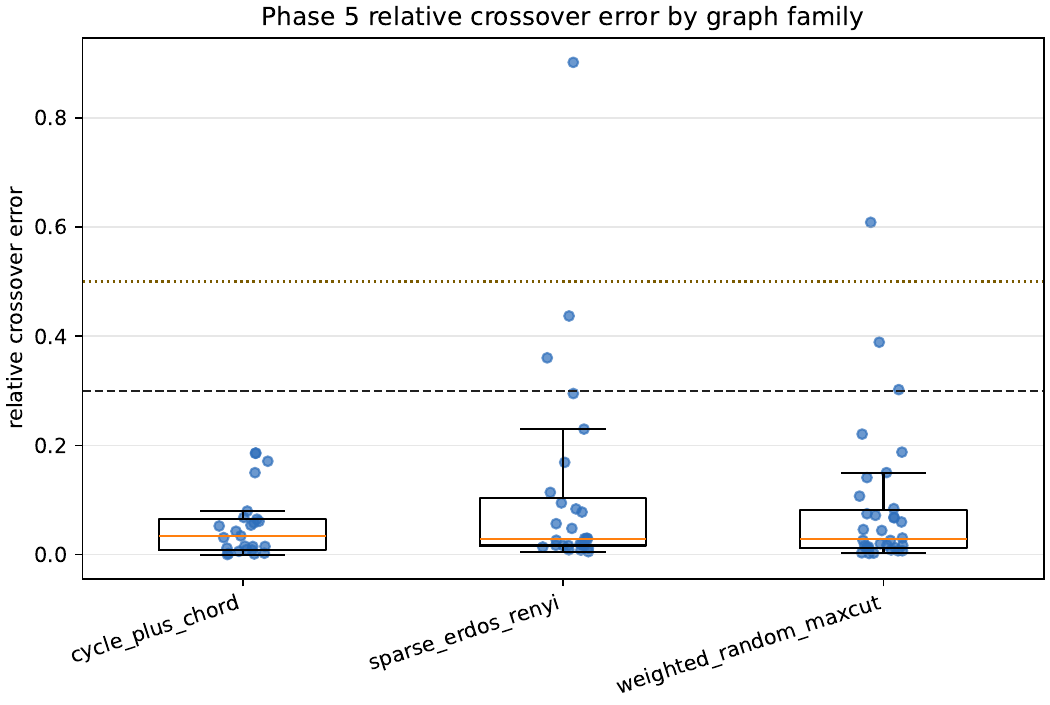}
\caption{Relative crossover error by graph family.  Most finite crossovers are
within the specified $<0.30$ relative-error threshold, and nearly all are within
$0.50$.}
\label{fig:phase5_relative_error_by_graph_family}
\end{figure}

\subsection{Pauli-channel robustness}
\label{subsec:pauli_channel_results}

The Pauli-channel robustness test verifies that the finite-shot prediction
extends beyond the symmetric depolarizing closed form for $\Gamma_{\mathcal Q}$.
The tested variants use positive channel-specific PEC inverses: N2a has
$p_x:p_y:p_z=2:1:1$, N2b uses $1:1:3$, and N3 adds sparse correlated $XX$ and
$ZZ$ two-qubit Pauli terms.

\begin{table}
\centering
\small
\caption{Pauli-channel robustness summary.  N2b and N3 have low-signal
points at $p=0.001$ and one boundary-local finite point each at
$p=0.002$; the overall conclusion remains consistent.}
\label{tab:phase6_pauli_summary}
\begin{ruledtabular}
\begin{tabular}{l c c c c}
\toprule
Model & Finite errors & Fraction $<0.30$ & Mean rel. error & Max rel. error \\
\colrule
N2a & 6 & 1.0 & 0.0553 & 0.1507 \\
N2b & 5 & 0.8 & 0.1103 & 0.4830 \\
N3  & 5 & 0.8 & 0.0840 & 0.3005 \\
\botrule
\end{tabular}
\end{ruledtabular}
\end{table}

\begin{figure}
\centering
\includegraphics[width=0.78\textwidth]{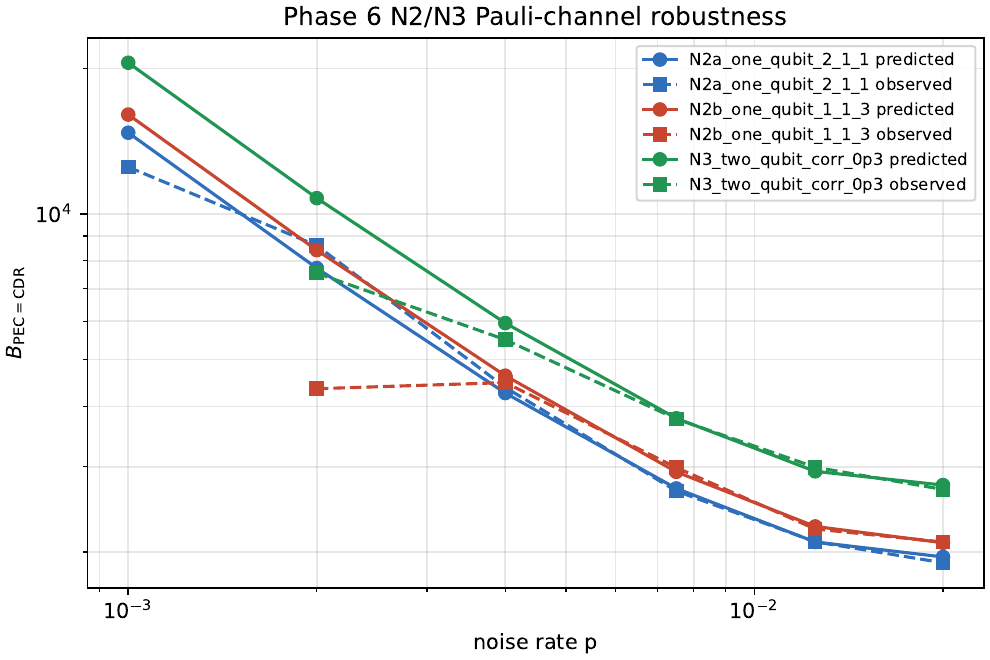}
\caption{Pauli-channel robustness of the PEC/CDR crossover prediction using
channel-specific PEC inverse one-norms.}
\label{fig:phase6_pauli_channel_robustness}
\end{figure}

\section{Secondary robustness tests}
\label{sec:phase7_supplement}

These secondary tests examine whether the baseline operating-window structure
persists under angle changes, a depth-two QAOA setting, near-zero baseline
response, and finite CDR training budgets.  They are reported after the main
baseline comparison because they probe robustness of the numerical picture
rather than define the baseline result.

\begin{table}
\centering
\footnotesize
\caption{Secondary robustness tests.  The selected perturbations preserve the
main operating-window interpretation while identifying where the predicted
crossover moves outside the sampled budget range.}
\label{tab:phase7_summary}
\begin{ruledtabular}
\begin{tabular}{@{}p{0.15\textwidth} p{0.36\textwidth} p{0.39\textwidth}@{}}
\toprule
Analysis & Numerical result & Physical interpretation \\
\colrule
Angle sensitivity &
All four selected angle pairs preserve the classification, and all non-boundary
classifications match. &
The three-region structure persists across the selected angle changes. \\
Depth-two QAOA &
The continuous-formula crossover
$B_{\PEC=\CDR}(p=0.005)\approx172$ lies below the selected grid minimum
$B_{\min}=256$, placing the crossover below the Monte Carlo grid. &
The depth-two setting therefore places the predicted crossover below the
sampled budget range for this grid. \\
Near-zero $\mu_0$ &
The qualitative structure is stable, with non-boundary agreement $0.992248$;
the CDR-dominant window moves outside the selected budget range in all
selected cases. &
Exact threshold formulas remain predictive, with the CDR-dominant window moving
outside the selected budget range near $\mu_0=0$, as explained by the
$1/\mu_0^2$ window displacement in Section~\ref{sec:phase_diagram}. \\
Finite training budget &
$B_{\mathrm{train,crit}}=4096$; classification agreement $0.975$ at the first
stable training budget; threshold margin $3.038$ standard errors. &
CDR retains a useful budget interval once the CDR training circuits receive a
sufficient finite shot budget. \\
\botrule
\end{tabular}
\end{ruledtabular}
\end{table}

\begin{figure}
\centering
\includegraphics[width=0.78\textwidth]{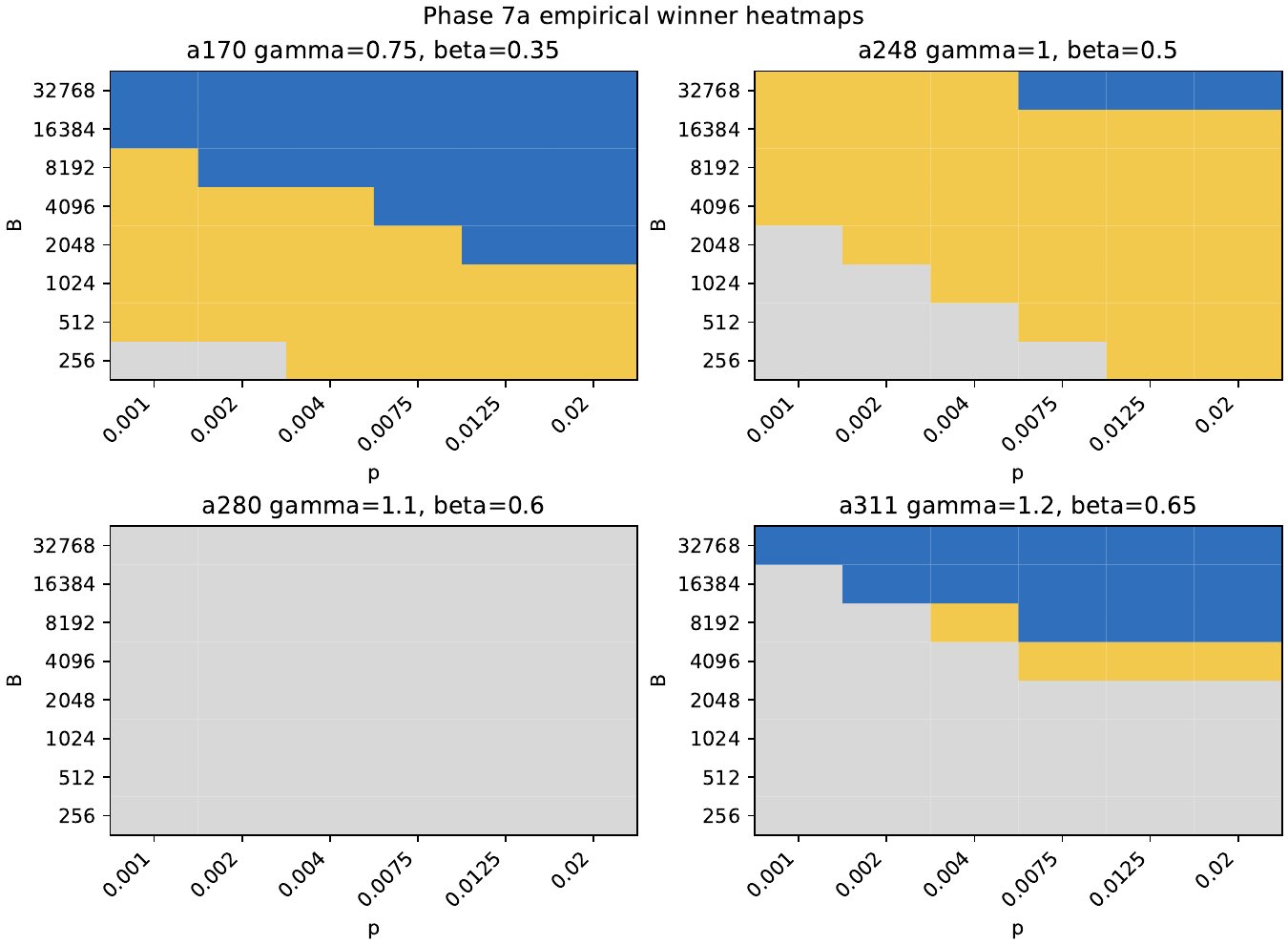}
\caption{Angle-sensitivity test across four selected QAOA angle pairs.}
\label{fig:phase7_angle_sensitivity}
\end{figure}

\begin{figure}
\centering
\includegraphics[width=0.78\textwidth]{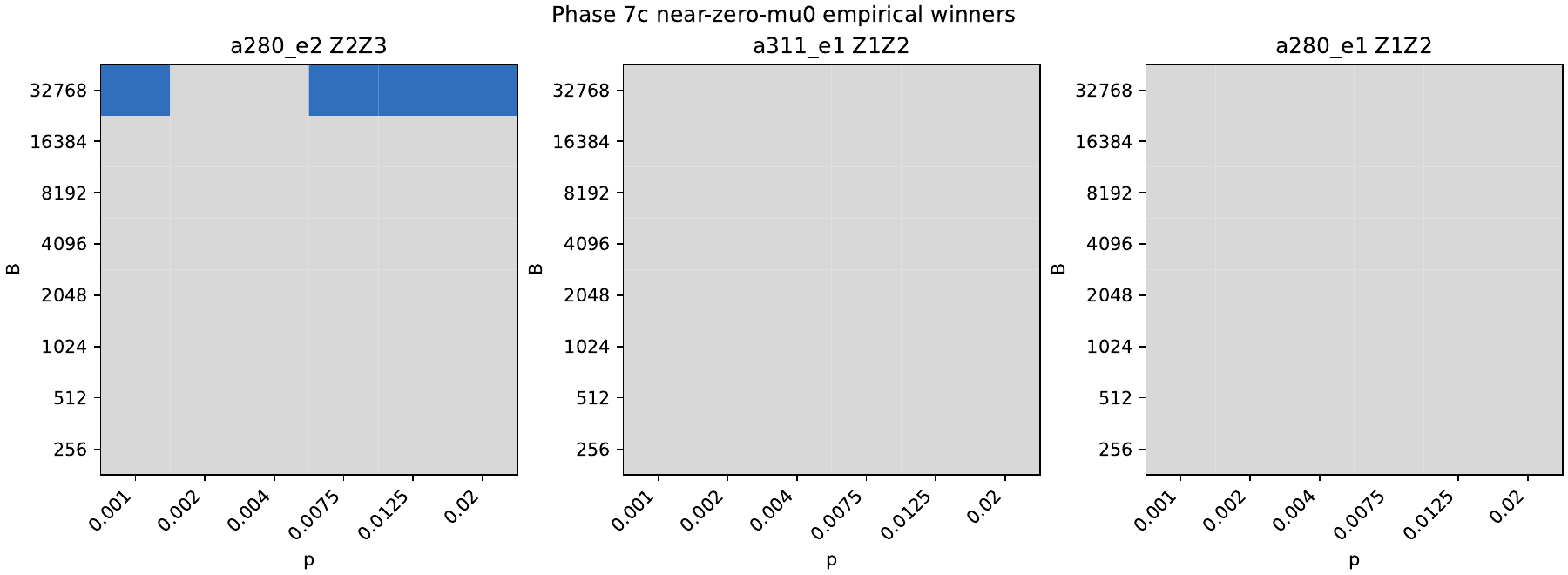}
\caption{Near-zero-$\mu_0$ stress test.  The exact formulas remain
predictive while the CDR-dominant region moves outside the selected budget range
in the near-zero cases.}
\label{fig:phase7_mu0_near_zero}
\end{figure}

\begin{figure}
\centering
\includegraphics[width=0.78\textwidth]{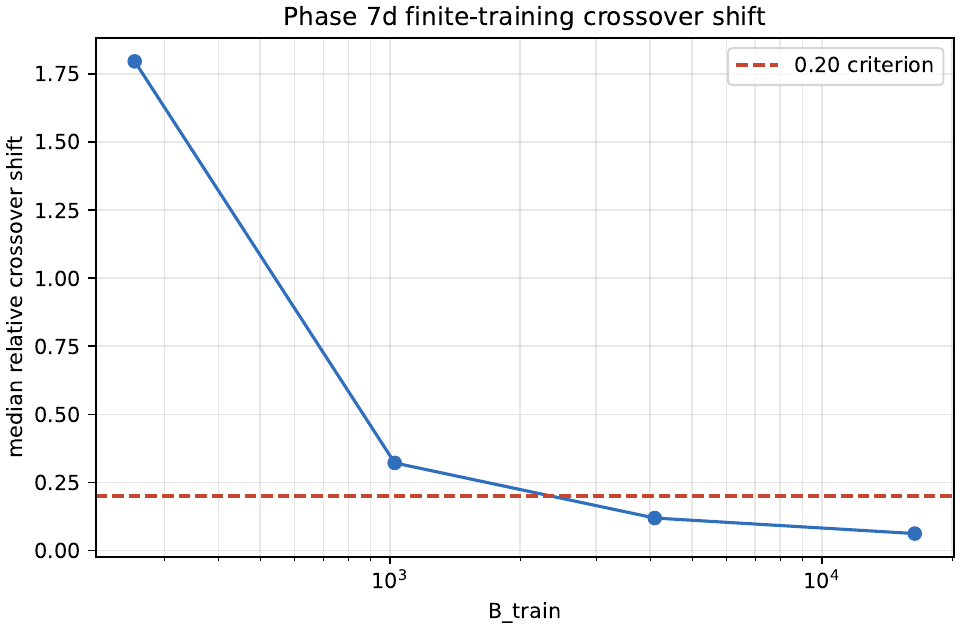}
\caption{Finite-training-shot test.  In the tested design,
$B_{\mathrm{train}}=4096$ is the first selected training budget with median
relative crossover shift below $0.20$.}
\label{fig:phase7_training_budget}
\end{figure}

\section{Calibration and finite-training tests}
\label{app:phase8_plus_experimental_layers}

The following tests address validation questions beyond the main baseline
comparisons.  The calibration hierarchy tests second-order CDR behavior,
response-projection floors, and finite-training shot allocation.  These analyses
are supplementary to the baseline operating-window result: they test how the
same finite-shot operating picture changes under calibrated first-order cancellation,
affine response projection, and finite allocation of shots between training and
target circuits.  Table~\ref{tab:phase8_to_10_summary} summarizes the
calibration and allocation outcomes.

\begin{table}
\centering
\scriptsize
\caption{Calibration and finite-training tests beyond the main
baseline comparisons.  The table summarizes observed outcomes for these
supplementary tests.}
\label{tab:phase8_to_10_summary}
\begin{ruledtabular}
\begin{tabular}{@{}p{0.12\textwidth} p{0.31\textwidth} p{0.47\textwidth}@{}}
\toprule
Analysis & Numerical result & Physical interpretation \\
\colrule
Second-order calibrated CDR &
The two root designs have maximum relative crossover errors
$0.1078$ and $0.0835$; their fitted slopes are $-2.855$ and $-2.938$, while
control slopes remain near $-1$. &
Supports the predicted second-order $p^{-3}$ crossover behavior when the
first-order calibration term is tuned near zero.  The test is an
oracle-assisted deterministic white-box density-matrix validation. \\
Response-projection floor &
The frozen instance ensemble has $21$ instances and $E_{\mathrm{aff}}=0.6521$; the
projection-oracle and T1 coefficient relative errors are $0.1125$ and $0.0762$.
The finite-grid lower-bound interval remains descriptive on this grid. &
Supports the frozen affine projection-floor interpretation; the finite-grid
interval is reported as a descriptive test for this ensemble. \\
Finite-training shot allocation &
The qualitative allocation trend is robust under the frozen design, while the
stricter optimum criterion is sensitive to finite-training resampling and the
regret test remains descriptive.  The total budget is $B_{\mathrm{tot}}=262144$ with
$N_t=48$ training circuits. &
Supports the qualitative shift toward larger training fractions for the
ill-conditioned design, while the stricter frozen optimum criterion is sensitive
to finite-training resampling. \\
\botrule
\end{tabular}
\end{ruledtabular}
\end{table}

\begin{figure}
\centering
\includegraphics[width=0.78\textwidth]{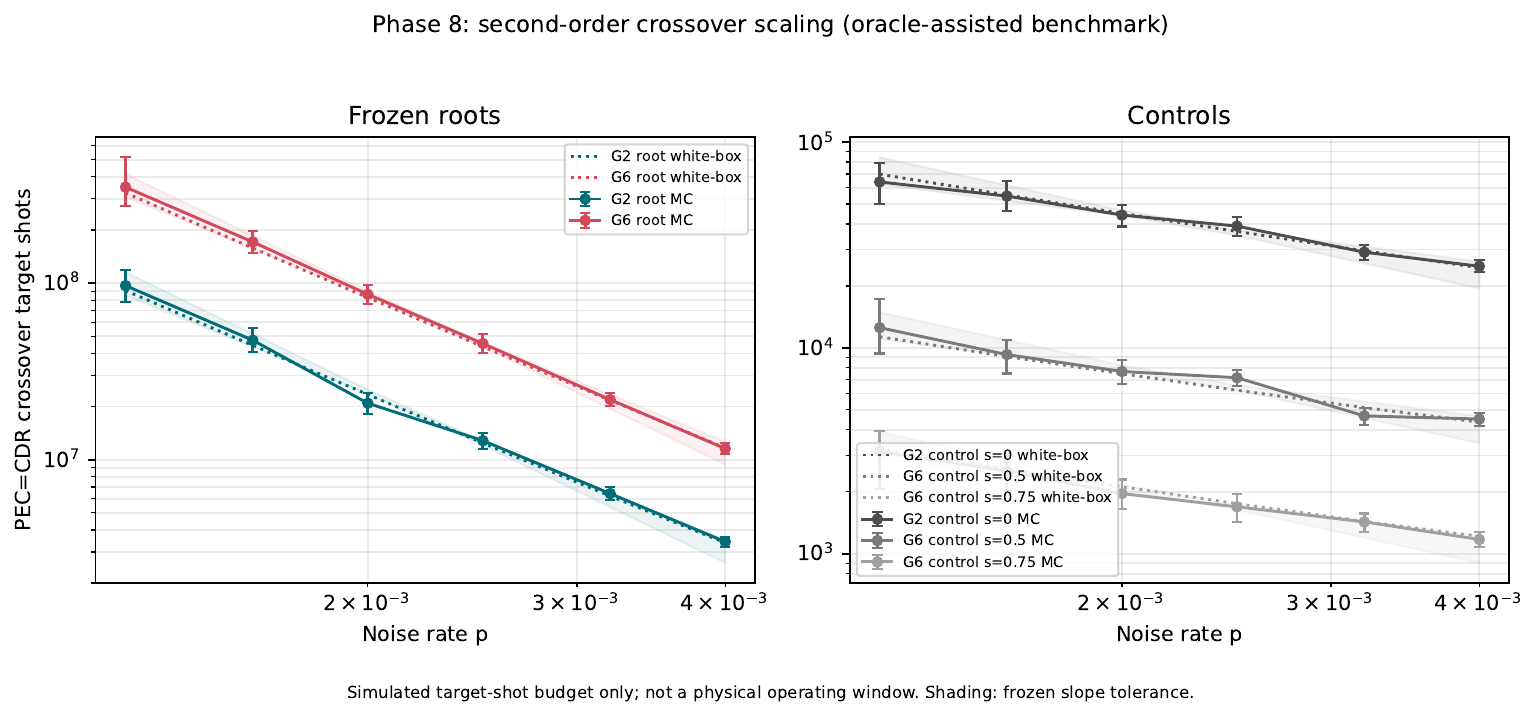}
\caption{Second-order crossover test.  Root designs probe the
first-order-calibrated regime, while controls retain the first-order scaling.}
\label{fig:phase8_second_order_scaling}
\end{figure}

\begin{figure}
\centering
\includegraphics[width=0.78\textwidth]{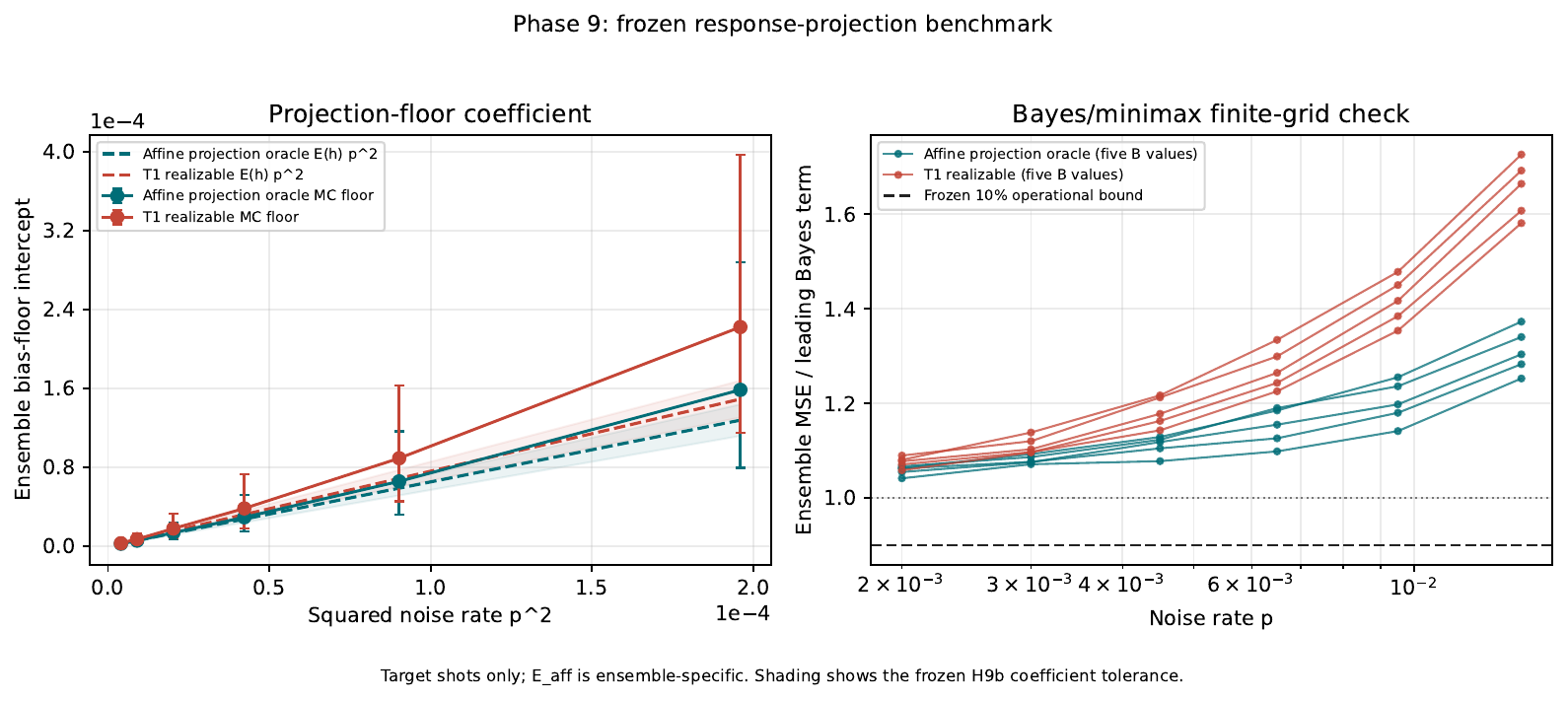}
\caption{Response-projection test over the frozen instance ensemble.}
\label{fig:phase9_projection_floor}
\end{figure}

\begin{figure}
\centering
\includegraphics[width=0.78\textwidth]{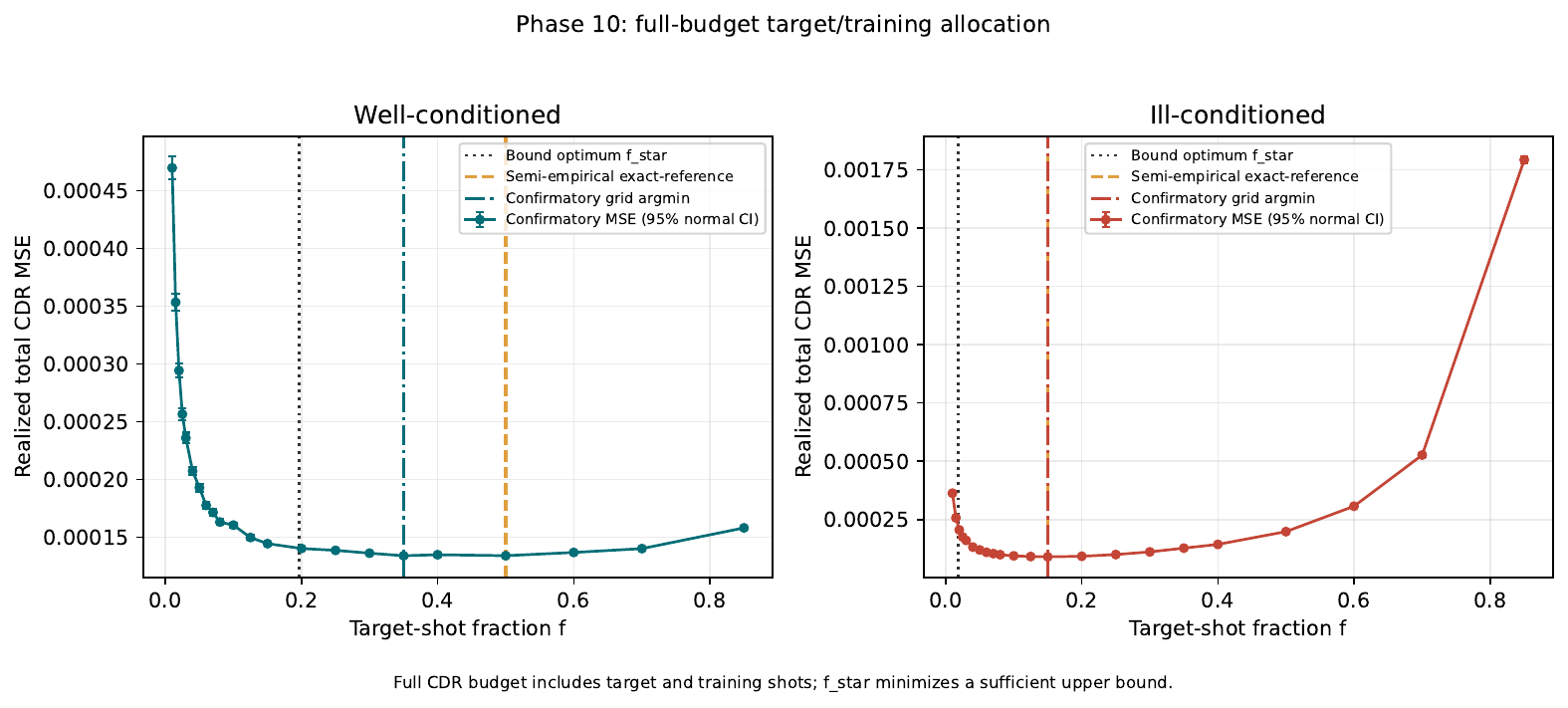}
\caption{Finite-training shot-allocation test under the full CDR
budget convention.}
\label{fig:phase10_allocation}
\end{figure}

\section{Device-model and archived hardware-count analyses}
\label{app:device_model_hardware_counts}

The device-facing analyses have two distinct roles.  First, offline
Qiskit Aer/FakeBackend simulations use device-derived noise and readout models
to test the finite-shot interpretation under controlled device-model
assumptions.  Second, archived real-device counts are reanalyzed from fixed
count data, with checksums used only to anchor reproducibility.  Table
\ref{tab:phase11_12_summary} summarizes these device-model and archived-count
outcomes.

\begin{table}
\centering
\scriptsize
\caption{Device-model and archived hardware-count analyses.  The offline
Aer/FakeBackend simulations test controlled device-model behavior, while the
archived-count analysis reports the hardware-facing evidence available from
fixed count data and checksums.}
\label{tab:phase11_12_summary}
\begin{ruledtabular}
\begin{tabular}{@{}p{0.12\textwidth} p{0.31\textwidth} p{0.47\textwidth}@{}}
\toprule
Analysis & Numerical result & Physical interpretation \\
\colrule
Device simulator map &
The offline Aer/FakeManilaV2 confirmatory simulator map has stable regions
consisting of unmitigated and CDR winners in this
configuration.  Accordingly, the specified positive PEC=CDR crossing lies
outside the observed stable map, and the surviving unmitigated=CDR boundary
tracks with maximum relative error $0.00732$. &
Shows that a device-derived sparse Pauli--Lindblad proxy can preserve the
CDR/no-mitigation operating boundary while removing the PEC-dominant region in
that configuration.  This is simulator evidence at the device-model layer. \\
Readout-bias analysis &
The maximum CDR residual change under common readout is $6.44\times10^{-5}$, and the CDR
slope-squared inflation prediction has maximum relative error $0.00662$. &
Supports the expected readout pattern: CDR is robust to shared physical readout
for the tested same-graph family, while PEC carries a nonzero readout component. \\
Model-violation analysis &
The PEC smaller-leading-bias ordering condition is not met in the tested setting.  Across tested
folds, the active crossing is the unmitigated=CDR boundary, which tracks with
maximum relative error $0.00330$; sensitivity tests cover readout maps,
FakeLimaV2, two-qubit-only folding, asymmetric readout, and resource
normalization. &
Shows a model-violation regime in which the CDR/no-mitigation boundary remains
the relevant operating boundary; the result is a robustness boundary case for
the PEC model-bias assumptions. \\
Archived hardware-count analysis &
Epochs 3, 4, and 6 form the analysis cluster; epoch 5 is retained only as an
age-of-calibration control.  The independent recomputation matches the archived
analysis with maximum absolute delta $0$, and the control difference
in mitigated
$|b_{\mathrm{PEC,res}}|$ is $0.1082$. &
Provides a reproducible hardware-count analysis with a cluster bootstrap over
archived epochs: the fixed-count analysis identifies CDR as the practical
finite-budget winner rather than a stable positive-boundary PEC improvement
over CDR.  The raw counts are archived at
\doi{10.5281/zenodo.20765589}; public reproduction uses those counts and their
checksums, and the result is reported as archival evidence. \\
\botrule
\end{tabular}
\end{ruledtabular}
\end{table}

\begin{figure}
\centering
\includegraphics[width=0.84\textwidth]{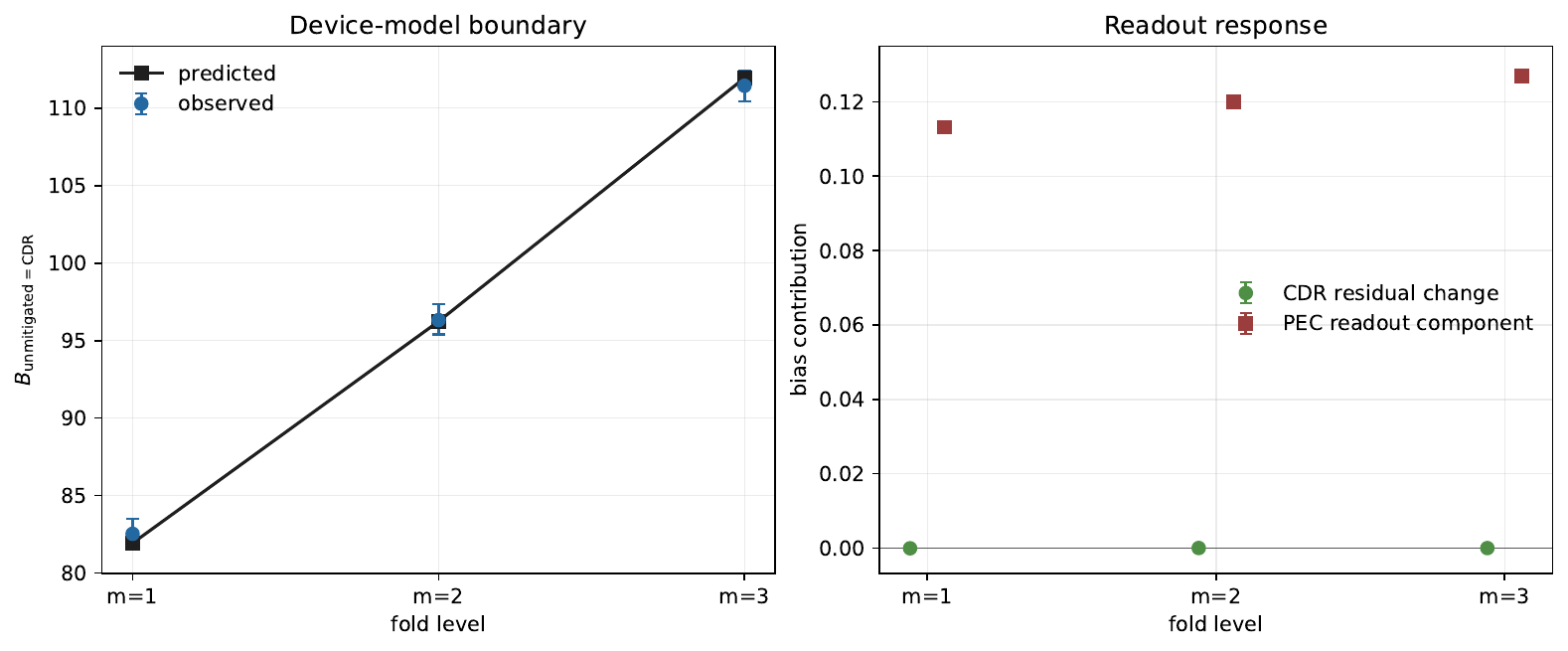}
\caption{Device-model simulator results.  Left: observed and predicted
unmitigated=CDR crossing in the offline Aer/FakeManilaV2 map.  Right:
common-readout response, showing near-invariant CDR residuals and a nonzero PEC
readout component.}
\label{fig:phase11_device_model_summary}
\end{figure}

For the archived readout/twirling contrast, the cluster mean absolute CDR
twirled-minus-mitigated bias gap is $4.02\times10^{-4}$, while the corresponding
PEC raw-minus-mitigated and twirled-minus-mitigated gaps are $2.87\times10^{-2}$
and $2.64\times10^{-3}$.

\begin{figure}
\centering
\includegraphics[width=0.95\textwidth]{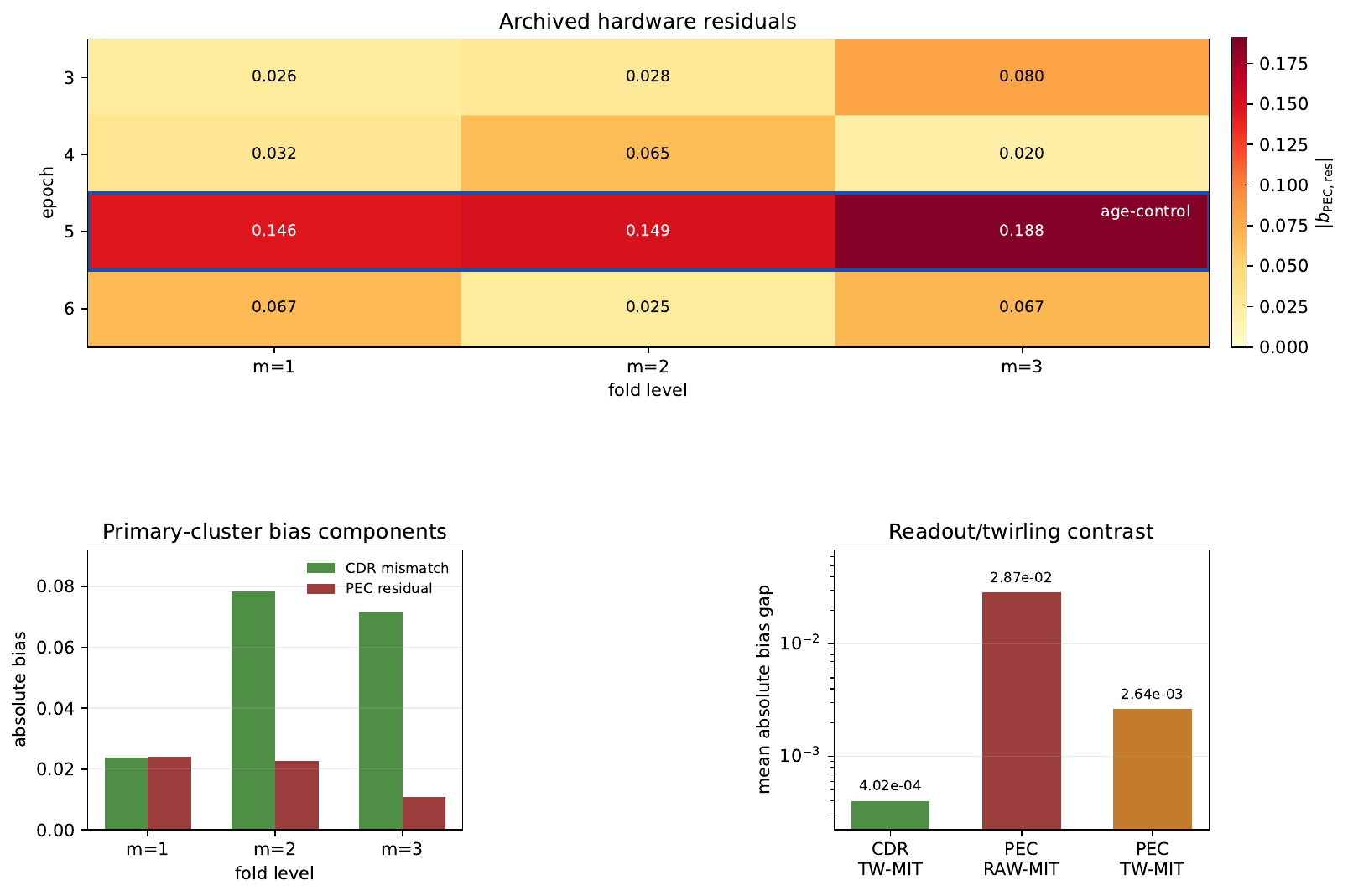}
\caption{Archived hardware-count analysis.  Top: mitigated PEC residual
magnitude across archived epochs and fold levels, with the age-control epoch
highlighted.  Bottom left: primary-cluster CDR mismatch and PEC residual
components.  Bottom right: cluster readout/twirling bias gaps.}
\label{fig:phase12_hardware_count_summary}
\end{figure}

\end{document}